\documentclass[10pt,onecolumn,draft]{IEEEtran2}
\usepackage{epsf,psfrag,amssymb,amsfonts,amsmath,cite}
\usepackage{graphicx}
\def\HOME{D:/Ge}
\newtheorem{theorem}{Theorem}
\newtheorem{example}{Example}
\newtheorem{corollary}{Corollary}
\newtheorem{lemma}{Lemma}
\newtheorem{definition}{Definition}

\newtheorem{proposition}{Proposition}

\def\psfancypar#1#2{\begingroup\def\par{\endgraf\endgroup\lineskiplimit=0pt}
               \setbox2=\hbox{\large\sc #2}
               \newdimen\tmpht \tmpht \ht2 \advance\tmpht by \baselineskip
               \font\hhuge=Times-Bold at \tmpht
               \setbox1=\hbox{{\hhuge #1}}
               \count7=\tmpht \count8=\ht1
               \divide\count8 by 1000 \divide\count7 by \count8
               \tmpht=.001\tmpht\multiply\tmpht by \count7
               \font\hhuge=Times-Bold at \tmpht
               \setbox1=\hbox{{\hhuge #1}}
               \noindent
                \hangindent1.05\wd1
               \hangafter=-2 {\hskip-\hangindent
               \lower1\ht1\hbox{\raise1.0\ht2\copy1}%
                \kern-0\wd1}\copy2\lineskiplimit=-1000pt}

\newcommand{\beq}{\begin{equation}}
\newcommand{\eeq}{\end{equation}}
\newcommand{\bqa}{\begin{eqnarray}}
\newcommand{\eqa}{\end{eqnarray}}
\newcommand{\bqn}{\begin{eqnarray*}}
\newcommand{\eqn}{\end{eqnarray*}}
\newcommand{\nn}{\nonumber}

\newcommand{\be}{\begin{enumerate}}
\newcommand{\ee}{\end{enumerate}}
\newcommand{\bi}{\begin{itemize}}
\newcommand{\ei}{\end{itemize}}
\newcommand{\bd}{\begin{description}}
\newcommand{\ed}{\end{description}}
\newcommand{\ba}{\begin{array}}
\newcommand{\ea}{\end{array}}
\newcommand{\bde}{\begin{definition}}
\newcommand{\ede}{\end{definition}}
\newcommand{\bex}{\begin{example}}
\newcommand{\eex}{\end{example}}

\def\boxit#1{\vbox{\hrule\hbox{\vrule\kern3pt
        \vbox{\kern3pt#1\kern3pt}\kern3pt\vrule}\hrule}}

\def\reals{ { {\rm  I \kern-0.15em R }  } }
\def\complex{ {\,{{\rm C} \kern-0.50em \raise0.20ex {  |}}\, }}

\def\0bf{{\bf 0}}
\def\1bf{{\bf 1}}
\def\2bf{{\bf 2}}
\def\3bf{{\bf 3}}
\def\4bf{{\bf 4}}
\def\5bf{{\bf 5}}
\def\6bf{{\bf 6}}
\def\7bf{{\bf 7}}
\def\8bf{{\bf 8}}
\def\9bf{{\bf 9}}

\def\dbf{{\bf d}}

\def\xbf{{\bf x}}
\def\ybf{{\bf y}}

\def\xbf{{\bf x}}
\def\ybf{{\bf y}}

\def\Dbf{{\bf D}}

\def\Rbf{{\bf R}}

\def\Xbf{{\bf X}}

\def\Dmat{\mathcal{D}}
\def\Emat{\mathcal{E}}

\def\Nmat{\mathcal{N}}

\def\Rmat{\mathcal{R}}

\def\Xmat{\mathcal{X}}

\def\Rxx{\Rbf_{\ssstyle X\kern-.1em X}}

\let\ssstyle=\scriptscriptstyle

\def\Kout{\setbox1=\hbox{\Huge\bf K}\hbox to
1.05\wd1{\hspace{.05\wd1}
}}
\def\Sout{\setbox1=\hbox{\Huge\bf S}\hbox to 1.05\wd1{\hspace{.05\wd1}
}}

\def\scalefig#1{\epsfxsize #1\textwidth}
\begin{document}
\title{\bf \LARGE Wyner's Common Information: Generalizations and A New Lossy Source Coding Interpretation}
\IEEEoverridecommandlockouts
\author{Ge Xu,  Wei Liu, ~\IEEEmembership{Student Member, ~IEEE,} and  Biao Chen, ~\IEEEmembership{Senior Member, ~IEEE}%
\thanks{G. Xu and B. Chen are with the Department of Electrical Engineering and Computer Science, Syracuse
University, Syracuse, NY 13244 USA (email: gexu@syr.edu, bichen@syr.edu).}%
\thanks{W. Liu was with the Department of Electrical Engineering and Computer Science, Syracuse
University, Syracuse, NY 13244 USA. He is now with Bloomberg L.P., New York, NY 10022 USA (email:wliusyr@gmail.com).}
}
\maketitle
\begin{abstract}
Wyner's common information was originally defined for a pair of dependent discrete random variables. Its significance is largely reflected in, hence also confined to, several existing interpretations in various source coding problems. This paper attempts to both generalize its definition and to expand its practical significance by providing a new operational interpretation. The generalization is two-folded: the number of dependent variables can be arbitrary, so are the alphabet of those random variables. New properties are determined for the generalized Wyner's common information of $N$ dependent variables. More importantly, a lossy source coding interpretation of Wyner's common information is developed using the Gray-Wyner network. In particular, it is established that the common information equals to the smallest common message rate when the total rate is arbitrarily close to the rate distortion function with joint decoding. A surprising observation is that such equality holds independent of the values of distortion constraints as long as the distortions are within some distortion region. Examples about the computation of common information are given, including that of a pair of dependent Gaussian random variables.
\let\thefootnote\relax\footnote{
 The material in this paper was presented in part at the Annual Allerton Conference on Communication, Control, and Computing, Montecillo, IL, Sept. 2010 and Annual Conference on Information Sciences and Systems, Baltimore, MD, March 2011. This work was supported in part by the Army Research Office under Award W911NF-12-1-0383, by the Air Force Office of Scientific Research under Award FA9550-10-1-0458, and by the National Science Foundation under Award 1218289.}
\end{abstract}
\section{Introduction}\label{section 1}
Consider  a pair of dependent random variables $X$ and $Y$ with
joint distribution $p(x,y)$,  which denotes either the probability density function if $X$ and $Y$ are continuous or the probability mass function if $X$ and $Y$ are discrete. Quantifying the information that is common between $X$ and $Y$ has been a classical problem both in information theory and in mathematical statistics\cite{Shannon:1948,Gacs_Korner_CI73,Ahlswede_Korner_CI74,Wyner_CI_75IT}. The most widely used notion is Shannon's mutual information, defined as
\[
I(X;Y)=E \left[\log \frac{p(x,y)}{p(x)p(y)} \right]
\]
where $p(x)$ and $p(y)$ are the marginal distribution of $X$ and $Y$ corresponding to the joint distribution $p(x,y)$ and $E[\cdot]$ denotes expectation with respect to $p(x,y)$. Shannon's mutual information measures the amount of uncertainty reduction in one variable by observing the other. Its significance lies in its applications to a broad range of problems in which concrete operational meanings of $I(X;Y)$ can be established. These include both source and channel coding problems in information and communication theory \cite{Cover:Book91} and hypothesis testing problems in statistical inference \cite{CK:book}.

Other notions of information have also been defined between a pair of dependent variables. Most notable among them are G\'{a}cs and  K\"{o}rner's
common randomness $K(X,Y)$ \cite{Gacs_Korner_CI73} and Wyner's common information $C(X,Y)$ \cite{Wyner_CI_75IT}. G\'{a}cs and  K\"{o}rner's common randomness is defined as the maximum number of common bits per symbol that can be independently extracted from $X$ and $Y$.
Quite naturally, $K(X,Y)$ has found extensive applications in secure communications, e.g., for key generation \cite{Maurer,Ahlswede&csiszar1,Ahlswede&csiszar2}.
More recently, a new interpretation of $K(X,Y)$ using the Gray-Wyner source coding network was given in \cite{Kamath_Anantharam}. It was noted in \cite{Gacs_Korner_CI73,Witsenhausen_CI75} that the definition of $K(X,Y)$ is rather restrictive in that $K(X,Y)$ equals $0$ in most cases except for the
special case
when $X=(X',V)$ and $Y=(Y',V)$ and
$X',Y',V$ are independent variables or those $(X,Y)$ pair that can be converted to such a dependence structure through relabeling the realizations, i.e., whose distribution is a permutation of the original joint distribution matrix. Notice also that $K(X,Y)$ is defined only for discrete random variables.

Wyner's common information was originally defined
for a pair of discrete random variables with finite alphabet as
\beq C(X,Y)=\inf_{X- W- Y} I(X,Y;W).\label{eq:CI}\eeq
Here, the infimum is taken over all auxiliary random variables $W$ such that
$X$, $W$, and $Y$ form a Markov chain. Clearly, the quantity $C(X,Y)$ in (\ref{eq:CI}) can be defined for any pair of random variables with arbitrary alphabets. However, the operational meanings of $C(X,Y)$ available in existing literature are largely confined to that for discrete $X$ and $Y$. These include the minimum common rate for the Gray-Wyner lossless source coding problem under a sum rate constraint, the minimum rate of a common input of two independent random channels for distribution approximation \cite{Wyner_CI_75IT}, and strong coordination capacity of a two-node network  without common randomness and with actions assigned at one node\cite{Cuff_10IT}.


This paper intends to generalize Wyner's common information along two directions. The first is to generalize it to that of multiple dependent random variables. The second is to generalize it to that of continuous random variables.

For the first direction, Wyner's common information is defined through a conditional independence structure which is equivalent to the Markov chain condition for two dependent variables. Relevant properties related to this generalization are derived. In addition, we prove that Wyner's original interpretations in \cite{Wyner_CI_75IT} can be directly extended to that involving multiple variables.  Note that both mutual information and common randomness have also been generalized to that of multiple random variables \cite{Hu62,Yeung:book,Tyagi_Narayan_Gupta_11IT}.

For the second direction, we provide a new lossy source coding interpretation using the Gray-Wyner network. Specifically, we show that, for the Gray-Wyner network, Wyner's common information is precisely the smallest common message rate for a certain range of distortion constraints when the total rate is arbitrarily close to the rate distortion function with joint decoding. As the common information is only a function of the joint distribution, this smallest common rate remains constant even if the distortion constraints vary, as long as they are in a specific distortion region. There has also been recent effort in characterizing the common message rate for lossy source coding using the Gray-Wyner network  \cite{Viswanatha&Akyol&Rose:ISIT12}. We establish the equivalence between the characterization in \cite{Viswanatha&Akyol&Rose:ISIT12} with an alternative characterization presented in the present paper.

Computing Wyner's common information is known to be a challenging problem; $C(X,Y)$ was only resolved for several special cases described in \cite{Wyner_CI_75IT, Witsenhausen_CI76}. Along with our generalizations of Wyner's common information, we provide two new examples where we can explicitly evaluate the common information of multiple dependent variables. In particular, we derive, through an estimation theoretic approach, $C(X,Y)$ for a bivariate Gaussian source and its extension to the multi-variate case with a certain correlation structure.

The rest of the paper is organized as follows.
Section \ref{section Gray Wyner} reviews Wyner's two approaches for the common information of two discrete random variables, the general Gray-Wyner network, and the relations among joint, marginal, and conditional rate distortion functions.
Section \ref{section CI}  gives the definition of Wyner's common information for $N$ dependent random variables with arbitrary alphabets along with some associated properties. The operational meanings of Wyner's common information developed in \cite{Wyner_CI_75IT} are  also extended to that of $N$ discrete dependent random variables in Section \ref{section CI}.  In Section \ref{section  CI for arbitrary alphabets}, we provide a new interpretation of Wyner's common information using Gray-Wyner's lossy source coding network. Specifically, we prove that for the Gray-Wyner network, Wyner's common information is precisely the smallest common message rate for a certain range of distortion constraints when the total rate is arbitrarily close to the rate distortion function with joint decoding.  In Section \ref{section 4}, two examples, the  doubly symmetric binary source
and the bivariate Gaussian source, are used to illustrate the lossy source coding interpretation of Wyner's common information.
The common information for bivariate Gaussian source and its extension to the multi-variate case  is also derived in \ref{section 4}.  Section \ref{section 5} concludes this paper.

{\em Notation:}
Throughout this paper, we use calligraphic letter $\Xmat$ to denote the alphabet and $p(x)$ to denote either point mass function
or probability density function of a random variable  $X$.
Boldface capital letter $\Xbf^A$ denotes a vector of random variables $\{X_i\}_{i\in A}$ where $A$ is an index set. $A\backslash B$ denotes set theoretic subtraction, i.e., $A\backslash B=\{x:x\in A \mbox{ and } x\notin B\}$. For two real vectors of identical size $\xbf$ and $\ybf$, $\xbf\leq \ybf$ denotes component-wise inequality.

\section{Existing Results}\label{section Gray Wyner}
\subsection{Wyner's result}\label{Wyner's result}

 \begin{figure}
\centerline{
\begin{psfrags}
\psfrag{xnyn}[c]{$X^n, Y^n$}
\psfrag{encoder}[c]{Encoder}
\psfrag{decoder1}[c]{Decoder 1}
\psfrag{decoder2}[c]{Decoder 2}
\psfrag{w1}[c]{$W_1$}
\psfrag{w2}[c]{$W_2$}
\psfrag{w0}[c]{$W_0$}
\psfrag{xnh}[c]{$\hat{X}^n$}
\psfrag{ynh}[c]{$\hat{Y}^n$}
 \scalefig{.45}\epsfbox{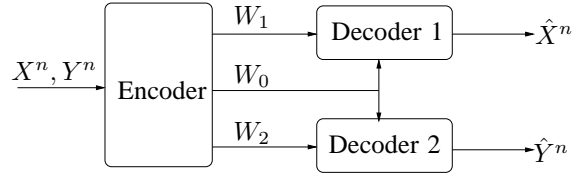}
\end{psfrags}}
\caption{\label{fig:model0}Source coding over a simple network.}
\end{figure}
\begin{figure}
\centerline{
\begin{psfrags}
\psfrag{W}[r]{$W$}
\psfrag{PROCESSOR 1}[l]{\!\!\!\!\small Processor 1}
\psfrag{PROCESSOR 2}[l]{\!\!\!\!\small Processor 2}
\psfrag{Xn}[c]{${\tilde X}^n$}
\psfrag{Yn}[c]{${\tilde Y}^n$}
 \scalefig{.3}\epsfbox{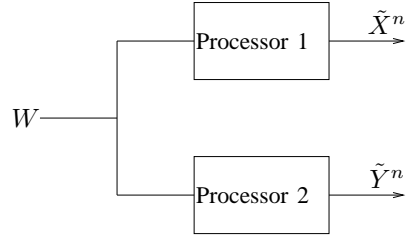}
\end{psfrags}}
\caption{\label{fig:model00}Random variable generators.}
\end{figure}
Wyner defined the common information of two discrete random variables $X$ and $Y$ with distribution $p(x,y)$ in equation (\ref{eq:CI}) and
provided two operational meanings for this definition.
 The first approach is shown in Fig.~\ref{fig:model0}. This model is a source coding network first studied by Gray and Wyner in \cite{Gray_Wyner_74}. In this model, the encoder observes a pair
 of sequences $(X^n,Y^n)$, and map
 them to three messages $W_0,W_1,W_2$, taking values in alphabets of respective sizes
 $2^{nR_0},2^{nR_1}$ and $2^{nR_2}$.
 Decoder 1, upon receiving $(W_0,W_1)$, needs to reproduce $X^n$ with high reliability
 while decoder 2, upon receiving $(W_0,W_2)$, needs to reproduce $Y^n$ with high reliability.
Define \bqn \Delta=\frac{1}{2n}\left(E[d_H(X^n,\hat{X}^n)]+E[d_H(Y^n,\hat{Y}^n)]\right)\eqn
where $d_H(\cdot,\cdot)$ is the Hamming distortion. Let $C_1$ be the the infimum of all achievable $R_0$ for the system in Fig.~\ref{fig:model0} such that for any $\epsilon>0$, there exists, for $n$ sufficiently large, a source code with
 the total rate  $R_0+R_1+R_2\leq H(X,Y)+\epsilon$  and $\Delta\leq \epsilon$.

The second approach is shown in Fig.~\ref{fig:model00}. In this approach, the joint distribution  $p(x^n,y^n)=\prod_{i=1}^np(x_i,y_i)$ is approximated
by the output distribution of a pair of  random number generators.
A common input $W$,
uniformly distributed on ${\cal W}=\{1,\cdots,2^{nR_0}\}$ is sent
to two separate processors which are  independent of
each other. These processors (random number generators)
generate independent and identically distributed (i.i.d) sequences
according to two distributions $q_1(x^n|w)$ and $q_2(y^n|w)$ respectively. The output
sequences of the two processors are denoted by ${\tilde X}^n$ and
${\tilde Y}^n$ respectively and the joint distribution of the
output sequences is given by
\bqn q(x^n,y^n)=\sum_{w\in{\cal W}}\frac{1}{{|\cal
W|}}q_1(x^n|w)q_2(y^n|w).\label{eq:sta}\eqn
Let  \bqn D_n(q,p)=\frac{1}{n} \sum_{x^n\in{\cal
X}^n,y^n\in {\cal
Y}^n}q(x^n,y^n)\log{\frac{q(x^n,y^n)}{p(x^n,y^n)}}.\eqn
Let $C_2$ be the infimum of rate $R_0$ for the common input such that for any $\epsilon>0$, there exists a pair of distribtions  $q_1(x^n|w)$, $q_2(y^n|w)$ and $n$  such that
$D_n(q,p)\leq \epsilon$.

Wyner proved in\cite{Wyner_CI_75IT} that
 \bqn
C_1=C_2= C(X,Y).
\eqn
\subsection{Generalized Gray-Wyner networks}

Consider the  Gray-Wyner source coding
network \cite{Gray_Wyner_74}  with one encoder and $N$ decoders as shown in Fig.~\ref{fig:model1}.
The encoder observes  an i.i.d. vector source sequence $\{\Xbf_1,\cdots,\Xbf_n\}$ where each $\Xbf_k=\{X_{1k},\cdots,X_{Nk}\}$, $k=1,\cdots, n$, is a length-$N$ vector with joint distribution $p(\xbf)$. Denote by $X_i^n=[X_{i1},\cdots, X_{in}]$ the $i$th component of the vector sequence. There are a total of $N$ receivers, with the $i$th receiver only interested in recovering the $i$th component sequence $X^n_i$.
 The encoder encodes the source into $N+1$ messages, one is a public message available at all receivers while the other $N$
messages are private messages only available at the corresponding receivers.

\begin{figure}
\centerline{
\begin{psfrags}
\psfrag{X1}[c]{$X_1^n, \cdots,X_N^n$}
\psfrag{Enc}[c]{Encoder}
\psfrag{D1}[c]{Decoder 1}
\psfrag{Dk}[c]{Decoder 2}
\psfrag{DK}[c]{Decoder N}
\psfrag{R1}[c]{$W_1$}
\psfrag{R0}[c]{$W_0$}
\psfrag{Rk}[c]{$W_2$}
\psfrag{RK}[c]{$W_N$}
\psfrag{2}[c]{$\cdots$}
\psfrag{hX1}[c]{$\hat{X}_1^n$}
\psfrag{hXk}[c]{$\hat{X}_2^n$}
\psfrag{hXK}[c]{$\hat{X}_N^n$}
 \scalefig{.5}\epsfbox{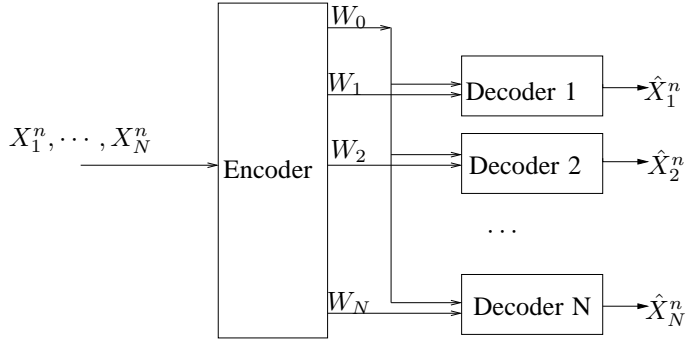}
\end{psfrags}}
\caption{\label{fig:model1} Generalized Gray-Wyner source coding network.}
\end{figure}

For $m=1,2,\cdots$, let $I_m=\{0,1,2,\cdots, m-1\}$. An $(n,M_0,M_1,\cdots,M_N)$ code is defined by
\begin{itemize} \item An encoder mapping
\[
 f:  {\cal X}^n_1 \times \cdots \times {\cal X}^n_N \rightarrow  I_{M_0}\times I_{M_1} \times \cdots I_{M_N},
 \]

\item $N$ decoder mappings
 \bqn g_i:
I_{M_i}\times I_{M_0} \rightarrow \hat{\cal X}^n_i, \hspace{0.2in} i=1,2,\cdots,N.\eqn
\end{itemize}
For an $(n,M_0,M_1,\cdots,M_N)$ code,
let $f(\Xbf_1,\cdots,\Xbf_n)=(W_0,W_1,\cdots,W_N)$ and $\hat{X}_i^n=g_i(W_i,W_0)$, $i=1,2,\cdots,N$.

We discuss below the lossless and lossy source coding using the generalized Gray-Wyner network.

\subsubsection{Lossless Gray-Wyner source coding}

$\\$
Define the probability of error as
\bqa  P^{(n)}_e= \frac{1}{nN}\sum_{i=1}^N E[ d_H(X_i^n,\hat{X}^n_i)],\eqa
where $\hat{X}^n_i=g_i(W_i,W_0)$
for $i=1,\cdots, N$ and $d_H(u^n,\hat{u}^n)$ is the Hamming distance between  $u^n$ and $\hat{u}^n$.

A rate tuple $(R_0,R_1,\cdots,R_N)$ is said to be {\em achievable} if for any $\epsilon>0$,  there exists, for $n$ sufficiently large, an  $(n,M_0,M_1,\cdots,M_N)$ code such that
\bqa M_i&\leq& 2^{n(R_i+\epsilon)},\hspace{0.2in} i=0,1,\cdots,N,\\
 P^{(n)}_e&\leq& \epsilon .\eqa

Denote by $\Rmat_1$ the region of all achievable rate tuples $(R_0,R_1,\cdots,R_N)$.
\begin{theorem}\label{them A1}
$\Rmat_1$
is the union of  all rate tuples $(R_0,R_1,\cdots,R_N)$ that satisfy
\bqa R_0&\geq& I(X_1,\cdots,X_N; W),\label{eq:them41}\\
R_i&\geq& H(X_i|W),\hspace{0.2in} i=1,2,\cdots,N,\label{eq:them42}\eqa
for some $W\sim p(w|x_1,\cdots,x_N)$. 
\end{theorem}
\subsubsection{Lossy Gray-Wyner source coding}

$\\$
Let $\dbf(\xbf,\hat{\xbf})\triangleq\{d_1(x_1,\hat{x}_1),\cdots,d_N(x_N,\hat{x}_N)\}$ be a compound distortion measure.
Define $\Delta_i, i=1,\cdots,N$ to be the average distortion between the $i$th component sequence of the encoder input and the $i$th decoder output,
\bqa \Delta_i=E[d_i(X^n_i,\hat{X}^n_i)]=\frac{1}{n}\sum_{k=1}^n E[d_i(X_{ik},\hat{X}_{ik})].\eqa

Define the vector of average distortions to be ${\bf \Delta}\triangleq\{\Delta_1,\cdots,\Delta_N\}$.
An $(n,M_0,M_1,\cdots,M_N)$ code with an average distortion vector ${\bf \Delta}$ is said to be an $(n,M_0,M_1,\cdots,M_N,{\bf \Delta})$ rate distortion code.
Let $\Dbf\triangleq \{D_1,D_2,\cdots, D_N\}\in \mathbb{R}^N_{+}$.  A rate tuple $(R_0,R_1,\cdots,R_N)$ is said to be {\em $\Dbf$-achievable} if for arbitrary $\epsilon>0$, there exists, for  $n$ sufficiently large, an  $(n,M_0,M_1,\cdots,M_N,{\bf \Delta})$ code  such that
\bqa M_i&\leq& 2^{n(R_i+\epsilon)},\hspace{0.2in} i=0,1,\cdots,N,\\
 {\bf \Delta}&\leq& \Dbf+\epsilon.\eqa
Let $\Rmat_2(\Dbf)$ be the region of all $\Dbf$-achievable rate tuples $(R_0,R_1,\cdots,R_N)$.
\begin{theorem}\label{them A2}
 $\Rmat_2(\Dbf)$ is the union of all rate tuples $(R_0,R_1,\cdots,R_N)$
 that satisfy
\bqa R_0&\geq& I(X_1,\cdots,X_N; W),\\
R_i&\geq& R_{X_i|W}(D_i),\hspace{0.2in} i=1,2,\cdots,N,\eqa
 for some $W\sim p(w|x_1,\cdots,x_N)$.
\end{theorem}

Here, $R_{X_i|W}(D_i)$ is the conditional rate distortion function  defined as \cite{Gray_72_CRD}
\beq R_{X_i|W}(D_i)=\min_{p_t(\hat{x}_i|x_i,w):Ed_i(X_i,\hat{X}_i)\leq D_i} I(X_i;\hat{X}_i|W).\label{eq CRD}\eeq

Theorems \ref{them A1} and \ref{them A2} are direct extensions of Theorem 4 and 8 in \cite{Gray_Wyner_74} for Gray-Wyner network with two receivers.
 Note that in \cite{Gray_Wyner_74}, the authors proved only the discrete case for \cite[Theorem 8]{Gray_Wyner_74}, the proof for continuous alphabets can be constructed in a similar fashion.

\subsection{Joint, marginal and conditional rate distortion functions}
In this section, we review the joint, marginal and conditional rate distortion functions and their relations.
Two-dimensional sources will be considered and the results can be generalized immediately to $N$-dimensional vector sources.

Given  a two-dimensional source $(X_1,X_2)$ with probability distribution $p(x_1,x_2)$ and two distortion measures $d_1(x_1,\hat{x}_1)$ and $d_2(x_2,\hat{x}_2)$
defined on  $\Xmat_1\times\hat{\Xmat}_1$ and  $\Xmat_2\times\hat{\Xmat}_2$,
the joint rate distortion function  is given by 
 \beq R_{X_1X_2}(D_1,D_2)=\min I(X_1X_2;\hat{X}_1\hat{X}_2),\eeq
 where the minimum is taken over all  test channels $p_t(\hat{x}_1\hat{x}_2|x_1x_2)$ such that $Ed_1(X_1,\hat{X}_1)\leq D_1$, $Ed_2(X_2,\hat{X}_2)\leq D_2$.
The conditional rate distortion function is defined in (\ref{eq CRD}).
The   joint, marginal and conditional rate distortion functions satisfy the following inequalities. 
\begin{lemma}\cite{Gray_73,Leiner_77}\label{lemma 1}
Given a two-dimensional source $(X_1,X_2)$ with joint distribution $p(x_1,x_2)$ and  two distortion measures $d_1(x_1,\hat{x}_1)$, $d_2(x_2,\hat{x}_2)$
defined respectively on $\Xmat_1\times\hat{\Xmat}_1$ and  $\Xmat_2\times\hat{\Xmat}_2$, the rate distortion functions satisfy the following inequalities
\begin{subequations}\label{all1}
\begin{align} R_{X_1X_2}(D_1,D_2)&\geq R_{X_1|X_2}(D_1)+R_{X_2}(D_2),\label{Rd 3}\\
  R_{X_1|X_2}(D_1)&\geq R_{X_1}(D_1)-I(X_1;X_2),\label{Rd 4}\\
 R_{X_1X_2}(D_1,D_2)&\geq R_{X_1}(D_1)+R_{X_1}(D_2)-I(X_1;X_2),\label{Rd 5}
 \end{align}\end{subequations}
\vspace{-0.4in}\begin{subequations}\label{all2}
\begin{align}
R_{X_1}(D_1)&\geq R_{X_1|X_2}(D_1),\label{Rd 1}\\
 R_{X_1}(D_1)+R_{X_2}(D_2)&\geq R_{X_1X_2}(D_1,D_2).\label{Rd 2}
 \end{align}\end{subequations}
Sufficient conditions for equality in (\ref{all1}) are that the optimum backward test channels for the functions on the left side of each
equation factor appropriately, i.e., for (\ref{Rd 3}) $p_b(x_1x_2|\hat{x}_1\hat{x}_2)=p(x_1|\hat{x}_1x_2)p(x_2|\hat{x}_2)$, for (\ref{Rd 4}) $p_b(x_1|\hat{x}_1x_2)=p(x_1|\hat{x}_1)$ and for (\ref{Rd 5}) that $p_b(x_1x_2|\hat{x}_1\hat{x}_2)=p(x_1|\hat{x}_1)p(x_2|\hat{x}_2)$.
Equalities hold in (\ref{all2}) if and only if $X_1$ and $X_2$ are independent.
\end{lemma}

 Furthermore, Gray has shown that under quite general conditions, equalities hold in (\ref{all1}) for small values of distortion. This is because the marginal, joint and conditional rate distortion functions
 equal to their  Extended Shannon Lower Bounds (ESLB) \cite{Gray_73},\cite{Gray_72_CRD} under suitable conditions. 
These ESLB, denoted by $R_X^{(L)}(D)$ for a rate distortion function $R_X(D)$, satisfy the following property.
Denote by $\Dmat$ a surface in the $m$-dimensional space and the inequality $\Delta\leq\Dmat$ means that there exists a vector $\beta\in\Dmat$ such that $\Delta\leq \beta$. If there is no such a vector,  $\Delta>\Dmat$. Likewise, $\Dmat_1\leq\Dmat_2$ means that $\beta\leq\Dmat_2$ for any $\beta\in\Dmat_1$ \cite{Gray_73}.
 \begin{lemma}\cite{Gray_73}\label{lemma 2}
 Given a two-dimensional source $(X_1,X_2)$ with joint distribution $p(x_1,x_2)$ such that for $x_1\in \Xmat_1,x_2\in \Xmat_2$, $p(x_2|x_1)>0$,
reproduction alphabets $\hat{\Xmat}_1=\Xmat_1$, $\hat{\Xmat_2}=\Xmat_2$ and two per-letter distortion measures $d_1(x_1,\hat{x}_1)$ and $ d_2(x_2,\hat{x}_2)$ that satisfy
\bqa d_i(x_i,\hat{x}_i)&>&d_i(x_i,x_i)=0, x_i\neq \hat{x}_i, i=1,2,\eqa
there exist strictly positive surfaces $\Dmat(X_1X_2)$, $\Dmat(X_1|X_2)$, $\Dmat(X_1)$ and $\Dmat(X_2)$ such that
\bqn \begin{array}{ll} R_{X_1X_2}(D_1,D_2)=R_{X_1X_2}^{(L)}(D_1,D_2), &\mbox{if}~(D_1,D_2)\leq \Dmat(X_1X_2),\\
R_{X_1|X_2}(D_1)=R_{X_1|X_2}^{(L)}(D_1), &\mbox{if}~ D_1\leq \Dmat(X_1|X_2),\\
R_{X_1}(D_1)=R_{X_1}^{(L)}(D_1), &\mbox{if}~D_1\leq \Dmat(X_1),\\
R_{X_2}(D_2)=R_{X_2}^{(L)}(D_2), &\mbox{if}~D_2\leq \Dmat(X_2),
\end{array}
\eqn
and
\bqn \Dmat(X_1|X_2)&\leq& \Dmat(X_1),\\
\Dmat(X_1X_2)&\leq& \left(\Dmat(X_1|X_2),\Dmat(X_2)\right) \leq \left(\Dmat(X_1),\Dmat(X_2)\right).\eqn
Finally,
\bqa R_{X_1X_2}^{(L)}(D_1,D_2)&=&R_{X_1|X_2}^{(L)}(D_1)+R_{X_2}^{(L)}(D_2),\label{ESLB 1}\\
&=&R_{X_1}^{(L)}(D_1)+R_{X_2}^{(L)}(D_2)-I(X_1;X_2)\label{ESLB 2}.\eqa
 \end{lemma}

It is apparent that when the rate distortion functions equal to their corresponding ESLB,
equations (\ref{ESLB 1}) and (\ref{ESLB 2}) imply equalities in (\ref{Rd 3}), (\ref{Rd 4}) and (\ref{Rd 5}).


\section{The Common Information of $N$ Dependent Discrete Random Variables}\label{section CI}
\subsection{Definition}
Wyner's original definition of the common information in (\ref{eq:CI}) assumes a Markov chain between the random variables $X$, $Y$ and the auxiliary variable $W$, i.e., $X-W-Y$. This Markov chain is equivalent to stating that $X$ and $Y$ are conditionally independent given $W$. This conditional independence structure can be naturally generalized to that of $N$ dependent random variables.
Let $\Xbf\triangleq\{X_1,\cdots,X_N\}$ be $N$ dependent random variables that take values
in some arbitrary (finite, countable, or continuous) spaces $\Xmat_1 \times \Xmat_2\times\cdots\times \Xmat_N$.
 The joint distribution of $\Xbf$ is denoted as $p(\xbf)$, which is either a probability mass function or a probability density function.
We now give the definition of the common information for $N$ dependent random variables.
\begin{definition}\label{definition CI}
Let $\Xbf$ be a random vector with  joint distribution
$p(\xbf)$. The common information of $\Xbf$ is defined as
\beq C(\Xbf)\triangleq \inf
I(\Xbf;W),\label{eq CI newdef}\eeq where the infimum is taken over all
the joint distributions of $(\Xbf,W)$ such that
\begin{enumerate}
\item the marginal distribution for $\Xbf$ is $p(\xbf)$, \label{theorem1.1}
\item  $\Xbf$ are conditionally independent given $W$, i.e.,
\bqa
p(\xbf|w)&=&\prod_{i=1}^{N}p(x_i|w).\label{theorem1.2} \eqa
\end{enumerate}
\end{definition}

\vspace{0.1in}

We now discuss several properties associated with the definition given in (\ref{eq CI newdef}).

Wyner's common information of two random variables $(X_1,X_2)$ satisfies the following inequality
  \bqn I(X_1,X_2)\leq C(X_1,X_2)\leq \min\{H(X_1),H(X_2)\}.\eqn
A similar inequality for the common information of $N$ random variables can be derived.
Let $A\subseteq\Nmat=\{1,2,\cdots,N\}$ and $\bar{A}=\Nmat\backslash A$.
We have
\bqa \max_A\{I(\Xbf^A;\Xbf^{\bar{A}})\}\leq C(\Xbf)\leq \min_j\{H(\Xbf^{-j})\},\eqa
where $\Xbf^{-j}\triangleq \Xbf^{\Nmat\backslash \{j\}}=$$\{X_1,\cdots,X_{j-1},X_{j+1},\cdots,X_N\}$ for $j\in\Nmat$.

To verify the upper bound, for any $j\in\Nmat$, let $W_j=\Xbf^{-j}$. Thus, $X_1,\cdots, X_N$ are conditionally independent given $W_j$,
and
\bqn I(\Xbf;W_j)=I(\Xbf;\Xbf^{-j})=H(\Xbf^{-j}).\eqn
Thus $C(\Xbf)\leq H(\Xbf^{-j})$ for all $j\in\Nmat$.

For the lower bound, since $X_1$, $\cdots$, $X_N$ are conditionally independent given $W$,
 we have the Markov chain
$\Xbf^A-W-\Xbf^{\bar{A}}$ for any subset $A\subseteq\Nmat$. Hence, 
\bqn I(\Xbf;W)\geq I(\Xbf^A;W)\geq I(\Xbf^A;\Xbf^{\bar{A}}),\eqn
where  the second inequality is by the data processing inequality.

Therefore,
\bqa I(\Xbf;W)\geq \max_A\{I(\Xbf^A;\Xbf^{\bar{A}})\}.\label{eq 28}\eqa

The common information defined in (\ref{eq CI newdef}) also satisfies the following monotone property.
\begin{lemma}
Let $\Xbf\sim p(\xbf)$. For any two sets
$A,B$ that satisfy $A\subseteq B \subseteq \Nmat=\{1,2,\cdots, N\} $, we have
 \beq
 C(\Xbf^{A})\leq C(\Xbf^{B}),
\eeq  
\end{lemma}

\begin{proof} Let $W'$ be the auxiliary variable that achieves $C(\Xbf^{B})$, i.e.,
$I(\Xbf^B;W')=\inf_W I(\Xbf^B;W)$. Since $A\subseteq B$, $\Xbf^B$ being
conditionally independent given $W'$ implies that $\Xbf^A$ are
conditionally independent given $W'$. Thus \bqn I(\Xbf^B;W')&\geq&
I(\Xbf^A;W'),\\
&\geq & \inf I(\Xbf^A;W),\eqn where the infimum is taken over all
$W$ such that $\Xbf^A$ is independent given $W$.\end{proof}

The above monotone property of the common information is contrary to what the name implies: conceptually, the information in common ought to decrease when new variables are included in the set of random variables. Such is the case for G\'{a}cs and K\"{o}rner's common randomness, i.e.,
$ K(\Xbf^{A})\geq K(\Xbf^{B})$.
As a consequence, we have that for any $N$ random variables
$ C(\Xbf)\geq K(\Xbf)$. The fact that the common information $C(\Xbf)$ increases as more variables are involved suggests that it may have potential
applications in statistical inference problems. This was explored in \cite{Xu_2011}.

\subsection{Coding theorems for the  common information of $N$ discrete random variables}

Section \ref{Wyner's result} describes
two operational interpretations of Wyner's common information for two discrete random variables based on the Gray-Wyner network and distribution approximation. These operational interpretations can also be extended to the common information of $N$ dependent random variables.

For the first approach, we consider the lossless Gray-Wyner network with $N$ terminals.
For the Gray-Wyner source coding network,
A number $R_0$ is said to be {\em achievable} if for any
$\epsilon>0$, there exists, for $n$ sufficiently large, an  $(n,M_0,M_1,\cdots,M_N)$ code
with
\bqa M_0&\leq& 2^{nR_0},\label{eq:achi0}\\
     \frac{1}{n}\sum^N_{i=0}\log{M_i}&\leq&H(\Xbf)+\epsilon,\label{eq:ach2}\\
     P^{(n)}_e&\leq& \epsilon.\label{eq:achi1}\eqa

Define $C_1$   as the
infimum of all achievable $R_0$.

\begin{theorem}\label{th 1}
For $N$ discrete random variables $\Xbf$  with joint distribution
$p(\xbf)$, 
\beq C_1=C(\Xbf).\eeq
\end{theorem}
The proof of Theorem \ref{th 1} is a direct extension  of the proof for two discrete random variables in \cite{Wyner_CI_75IT} and hence is omitted.

The second approach of interpreting the common information of discrete random variable uses distribution approximation.
Let $\{\Xbf_1,\cdots, \Xbf_n\}$ be i.i.d. copies of $\Xbf$ with distribution $p(\xbf)$, i.e., the joint distribution for  $\{\Xbf_1,\cdots, \Xbf_n\}$ is 
 \beq p^{(n)}(\xbf_1,\cdots,\xbf_n)=\prod^n_{k=1}
p(\xbf_k).\label{eq:pn}\eeq
An $(n, M, \Delta) $ {\em generator} consists of the
following:
\bi
 \item a message set ${\cal W}$ with cardinality $M$;
 \item for all $w\in {\cal W}$, probability distributions $q^{(n)}_i(x^n_i|w)$, for
$i=1,2,\cdots,N$.
\ei
Define the probability distribution on $ {\cal
X}^n_1 \times {\cal X}^n_2 \times \cdots \times {\cal X}^n_N$ \beq
q^{(n)}(\xbf_1,\cdots,\xbf_n)=\sum_{w\in{\cal
W}}\frac{1}{M}\prod^N_{i=1}
q^{(n)}_i(x^n_i|w).\label{eq:qn}\eeq

Let \beq
\Delta=D_n\left(q^{(n)}(\xbf_1,\cdots,\xbf_n);p^{(n)}(\xbf_1,\cdots,\xbf_n)\right)=\sum_{\xbf^n}\frac{1}{n}q^{(n)}(\xbf_1,\cdots,\xbf_n)
\log{\frac{q^{(n)}(\xbf_1,\cdots,\xbf_n)}{p^{(n)}(\xbf_1,\cdots,\xbf_n)}},
\eeq where   $p^{(n)}(\xbf_1,\cdots,\xbf_n)$ is defined in (\ref{eq:pn}) and $q^{(n)}(\xbf_1,\cdots,\xbf_n)$ is defined as in (\ref{eq:qn}).

A number $R$ is said to be {\em achievable} if for all $\epsilon>0$, if for $n$ sufficiently large
 there exists an $(n,M,\Delta)$ generator with $M\leq 2^{nR}$ and $\Delta \leq \epsilon$.
Define $C_2$ as the infimum of all achievable $R$.

\begin{theorem}\label{th 2}
For $N$ discrete random variables $\Xbf$  with joint distribution
$p(\xbf)$, 
\beq C_2=C(\Xbf).\eeq
\end{theorem}

The proof
can be constructed in the same way as that of  \cite[Theorems 5.2 and 6.2]{Wyner_CI_75IT}, hence is omitted.

\section{The Lossy Source Coding Interpretation of Wyner's Common Information}\label{section  CI for arbitrary alphabets}
The common information defined in (\ref{eq:CI}) and (\ref{eq CI newdef}) equally applies to that of continuous random variables. However, such definitions are only meaningful when they are associated with concrete operational interpretations. In this section, we develop a
lossy source coding interpretation of Wyner's common information using the Gray-Wyner network. While this new interpretation holds for the general case of $N$ dependent random variable, we elect to present coding theorems involving only a pair of dependent variables for ease of notion and presentation.


\subsection{Lossy Gray-Wyner source coding}

Given a two-dimensional source $(X_1,X_2)\sim p(x_1,x_2)$,  for any $(D_1,D_2)\geq 0$, a number $R_0$ is said to be \emph{$(D_1,D_2)$-achievable} if for any $\epsilon>0$, there exists, for $n$ sufficiently large, an $(n,M_0,M_1,M_2,\Delta_1,\Delta_2)$ code with
 \bqa M_0&\leq &2^{nR_0},\label{CI EQ1}\\
 \sum_{i=0}^2\frac{1}{n}\log M_i&\leq& R_{X_1X_2}(D_1,D_2)+\epsilon,\label{CI EQ2}\\
 \Delta_1\leq D_1+\epsilon &,& \Delta_2\leq D_2+\epsilon.\label{CI EQ3}\eqa

Define
$C_3(D_1,D_2)$  as the infimum of  all $R_0$'s that are  $(D_1, D_2)$-achievable. Thus, $C_3(D_1,D_2)$ is the minimum  common message  rate for the Gray-Wyner network with sum rate $R_{X_1X_2}(D_1,D_2)$ while satisfying the distortion constraint. Since  $R_{X_1X_2}(D_1,D_2)$ is always $(D_1, D_2)$-achievable, it is obvious that \bqa C_3(D_1,D_2)\leq R_{X_1X_2}(D_1,D_2).\eqa
The following theorem gives a precise characterization of $C_3(D_1,D_2)$.


\begin{theorem}\label{th 3}
\bqa C_3(D_1,D_2)=\tilde{C}(D_1,D_2),\eqa
where $ \tilde{C}(D_1,D_2)$ is the solution of the following optimization problem: \bqa &\inf&I(X_1,X_2;W)\label{eq 13}\\
&\mbox{subject to}&R_{X_1|W}(D_1)+R_{X_2|W}(D_2) +I(X_1,X_2;W)=R_{X_1X_2}(D_1,D_2).\nn\eqa
\end{theorem}
\begin{proof}
See Appendix \ref{appendix 3}.
\end{proof}

The authors in \cite{Viswanatha&Akyol&Rose:ISIT12} gave an alternative characterization of $C_3(D_1,D_2)$.  Define
\bqn C^*(D_1,D_2)=\inf I(X_1,X_2;W),
\eqn
where the infimum is taken over all joint distributions for $X_1,X_2,X^*_1,X^*_2,W$
such that
\bqa X_1^*-W-X^*_2\label{eq 1},\\
(X_1,X_2)-(X^*_1,X^*_2)-W\label{eq 2},\eqa
 where $(X^*_1,X^*_2)$ achieves $R_{X_1X_2}(D_1,D_2)$.
It was shown in \cite{Viswanatha&Akyol&Rose:ISIT12} that $C_3(D_1,D_2)=C^*(D_1,D_2)$. This, combined with Theorem \ref{th 3}, establishes that
\bqa \tilde{C}(D_1,D_2)=C^*(D_1,D_2).\label{eq 34}\eqa
$\tilde{C}(D_1,D_2)$ is derived from the rate distortion region $\Rmat_2(D_1,D_2)$  given in Theorem  \ref{them A2} while the authors in \cite{Viswanatha&Akyol&Rose:ISIT12} chose to derive  $C^*(D_1,D_2)$ from an alternative characterization of  $\Rmat_2(D_1,D_2)$ given in \cite{Venkataramani_Kramer_Goyal_03IT}.
In Appendix \ref{appendix 3-0}, we provide a direct proof of (\ref{eq 34}) for completeness. Also, as given in Appendix \ref{appendix 3-0}, a necessary condition for  the equality condition in the optimization problem (\ref{eq 13})  is \bqn R_{X_1X_2|W}(D_1,D_2)=R_{X_1|W}(D_1)+R_{X_2|W}(D_2).\eqn

\subsection{The relation of $C_3(D_1,D_2)$ and the common information}
Given our characterization of $C_3(D_1,D_2)$ in Theorem \ref{th 3}, we now establish its connection with $C(X_1,X_2)$ which leads to a new interpretation of Wyner's common information. We begin with the following two lemmas.
\begin{lemma}\label{lemma 3}
Let  $W$ be the random variable that achieves the common information of $X_1$ and $X_2$. If
\bqn  R_{X_1X_2|W}(D_1,D_2)+C(X_1,X_2)=R_{X_1X_2}(D_1,D_2),\eqn
then
\bqa C_3(D_1,D_2)\leq C(X_1,X_2)\label{eq:C3C}.\eqa
\end{lemma}

Lemma \ref{lemma 3} is a direct consequence of Theorem \ref{th 3} as the Markov chain $X_1-W-X_2$ implies $R_{X_1X_2|W}(D_1,D_2)=R_{X_1|W}(D_1)+R_{X_2|W}(D_2)$. Thus, the equality constraint in (\ref{eq 13}) is satisfied. Inequality (\ref{eq:C3C}) follows as
\[
\tilde{C}(D_1,D_2)=C_3(D_1,D_2)\leq I(X_1,X_2;W)=C(X_1,X_2).
\]

The next lemma gives a sufficient condition under which $C_3(D_1,D_2)\geq C(X_1,X_2)$ is true.
\begin{lemma}\label{lemma 4-0}
For any distortion pair $(D_1,D_2)$, if the rate distortion function satisfies
\bqa R_{X_1X_2}(D_1,D_2)=R_{X_1}(D_1)+R_{X_2}(D_1)-I(X_1;X_2),\label{eq lemma 4-0}\eqa
then we have
\bqn C_3(D_1,D_2)\geq C(X_1,X_2).\eqn
\end{lemma}
\begin{proof}See Appendix \ref{Appendix 4-0}.
\end{proof}

Lemmas \ref{lemma 3} and \ref{lemma 4-0}, together with the relations of marginal, joint and conditional rate distortion functions described in Lemmas \ref{lemma 1} and \ref{lemma 2}, allow us to determine a region
such that $C_3(D_1,D_2)$ equals to the common information.

\begin{theorem}\label{th 4}
Let  random variables $X_1,X_2$ be distributed as
$p(x_1,x_2)$ on $\Xmat_1\times \Xmat_2$
such that for $x_1\in \Xmat_1,x_2\in \Xmat_2$, $p(x_2|x_1)>0$.
Let the reproduction alphabets $\hat{\Xmat}_1=\Xmat_1$, $\hat{\Xmat_2}=\Xmat_2$.
The two per-letter distortion measures $d_1(x_1,\hat{x}_1)$, $d_2(x_2,\hat{x}_2)$ satisfy
\bqa d_i(x_i,\hat{x}_i)&>&d_i(x_i,x_i)=0, \quad x_i\neq \hat{x}_i, i=1,2.\eqa
Then there exists a strictly positive surface $\gamma\triangleq (\gamma_1,\gamma_2)$ such that, for $(D_1, D_2)\leq \gamma$,
\beq C_3(D_1,D_2)= C(X_1,X_2).\eeq
\end{theorem}
\begin{proof} See Appendix \ref{Appendix 4}.
\end{proof}

 Theorem \ref{th 4} shows that Wyner's common information is precisely the smallest common message rate $C_3(D_1,D_2)$  of Gray-Wyner network for a certain range of distortion constraints when the total rate is arbitrarily close to the rate distortion function with joint decoding.  As the common information is only a function of the joint distribution, hence is a constant for a given $p(x_1,x_2)$,  it is surprising that the smallest common rate $C_3(D_1,D_2)$ remains constant even if the distortion constraints vary, as long as they are in a specific distortion region.

While Theorem \ref{th 4} establishes that $C_3(D_1,D_2)=C(X_1,X_2)$ for $(D_1,D_2)\leq \gamma$,  it does not specify the value of the positive distortion vector $\gamma$.
Let $\Dmat^c\triangleq (D^c_1,D^c_2)$ be the two-dimensional distortion surface such that $R_{X_1X_2}(D^c_1,D^c_2)=C(X_1,X_2)$, then we must have
that $\gamma\leq\Dmat^c$. This is because if $\gamma>\Dmat^c$, then there exists $(D_1,D_2)$ such that $\gamma\geq (D_1,D_2)>\Dmat^c$ and
$C_3(D_1,D_2)\leq R_{X_1X_2}(D_1,D_2)<R_{X_1X_2}(D^c_1,D^c_2)=C(X_1,X_2)$, which  contradicts Theorem \ref{th 4}.
Now let us consider a particular point on the surface $\Dmat^c$. 
Let $W$ be the auxiliary random variable that achieves $C(X_1,X_2)$. Suppose there exists a distortion pair $(D_1^0,D_2^0)$ satisfying, for $i=1,2$, 
\bqa \begin{array}{rll} R_{X_i}(D_i^0)&=&I(X_i;W),\\
 D_i^0&=&\inf_{\hat{x}_i(w)}Ed_i(X_i,\hat{X}_i^0(W)),\end{array}\label{ratedistortion}\eqa
 where $\hat{x}_1^0(w),\hat{x}_2^0(w)$ are deterministic functions. Under this assumption, we can show that  $R_{X_1X_2}(D^0_1,D^0_2)=I(X_1,X_2;W)$.
Therefore, the joint rate distortion function  $R_{X_1X_2}(D^0_1,D^0_2)$  not only equals to the common information but also is achieved by the auxiliary random variable $W$.
Furthermore, it is easy to show  \bqa C_3(D_1^0,D_2^0)=C(X_1,X_2),\eqa using Lemma \ref{lemma 4-0} and the fact that $C_3(D_1^0,D_2^0)\leq R_{X_1X_2}(D^0_1,D^0_2)$. 
This means that in the Gray-Wyner network,  with the total rate equal to $R_{X_1X_2}(D_1^0,D_2^0)$, the scheme to transmit the pair of sources $(X_1^n,X^n_2)$  within distortion constraints $(D_1^0,D_2^0)$  is to communicate $W$ to the two receivers  using the common channel.

Let us now decrease the distortion constraints from $(D_1^0,D_2^0)$ to $(D_1,D_2)\leq (D_1^0,D_2^0)$. The question is whether the rate $C(X_1,X_2)$ is $(D_1,D_2)-$achieveble, i.e., if it is possible to transmit the sources $(X_1^n,X^n_2)$  with smaller distortions $(D_1,D_2)$ with the sum rate  at $R_{X_1X_2}(D_1,D_2)$   while keeping  the common rate at $C(X_1,X_2)$.
 In the following, we identify a sufficient condition for $C_3(D_1,D_2)=C(X_1,X_2)$ for successively refinable sources. A source $X$ with distortion measure $d(x,\hat{x})$ is said to be successively refinable from a coarser distortion $\delta_1$ to a finer distortion $\delta_2$ ($\delta_1\geq\delta_2$) if it can be encoded in two stages in which the optimal descriptions at the second stage is a refinement of the optimal descriptions at the first stage\cite{Equitz_91IT}. Similar definition can be applied to vector sources with individual distortion constraints and the details can be found in \cite{Nayak_10IT}.

In the following theorem, we  give a sufficient condition under which $ C_3(D_1,D_2)= C(X_1,X_2)$ for any $(D_1,D_2)\leq (D_1^0,D_2^0)$. This sufficient condition ensures the optimality of a two-stage encoding scheme: first encode the common message with rate $C(X_1,X_2)$ and we can obtain a coarse distortion $(D_1^0,D_2^0)$, then encode the two private messages with rates $R_{X_1|W}(D_1)$ and $R_{X_2|W}(D_2)$. The successive refinement assumption guarantees that the two-step approach can achieve the distortion $(D_1,D_2)$ and the sum rate  does not exceed the total rate  $R_{X_1X_2}(D_1,D_2)$.

\begin{theorem}\label{theorem 9}
Let $W$ be the auxiliary variable that  achieves $C(X_1,X_2)$ and $(D_1^0,D_2^0)$ be a distortion pair satisfying  (\ref{ratedistortion}).
If the source $(X_1,X_2)$ is successively refinable from $(D_1^0,D_2^0)$ to $(D_1,D_2)$ for $(D_1,D_2)\leq (D_1^0,D_2^0)$, and $X_i$ is successively refinable from $D_i^0$ to $D_i$ for $D_i\leq D_i^0$, $i=1,2$,  
then,  \bqn C_3(D_1,D_2)= C(X_1,X_2).\eqn

 \end{theorem}
\begin{proof} See Appendix \ref{Appendix them 9}.\end{proof}

In the following section, we will consider two examples involving successively refinable sources: the binary random variables and bivariate Gaussian variables. For these two cases, we compute explicitly the function $C_3(D_1,D_2)$ and establish its connection with $C(X_1,X_2)$. The distortion pair $(D_1^0,D_2^0)$  satisfying (\ref{ratedistortion}) are identified for both cases, thus Theorem \ref{theorem 9} can be directly applied.

\section{Examples}\label{section 4}

\subsection{Binary random variables}
Let $S\sim \mbox{Bern}(\theta)$ for $0\leq \theta\leq 1$, i.e., $S\in\{0,1\}$ and $P(S=1)=\theta$. Let $X_i$, $i=1,\cdots,N$, be the output of a binary symmetric channel (BSC) with crossover probability $a_1$ $(0\leq a_1\leq \frac{1}{2})$ and with $S$ as input. The BSC channels are independent of each other. Thus,
\bqn p(x_1,\cdots, x_N|s)=\prod_{i=1}^Np(x_i|s),\eqn
where
\bqn p(x_i|s)=\left\{\begin{array}{ll}
1-a_1, &\mbox{if $x_i=s$},\\
a_1, &\mbox {otherwise},\\
\end{array}\right.
\eqn
 for  $x_i\in\{0,1\}$. Therefore, the joint distribution of $X_1,X_2,\cdots,X_N$ is
\bqa p(x_1,x_2,\cdots,x_N)&=&\sum_{s\in\{0,1\}}p(s)\prod_{i=1}^Np(x_i|s),\nn\\
&=&\theta a_1^{t_N}(1-a_1)^{N-t_N}+(1-\theta)(1-a_1)^{t_N}a_1^{N-t_N},\label{eq 12}\eqa

where $t_N=\sum_{i=1}^Nx_i$.

For $N=2$, the joint distribution of $X_1,X_2$ is given by the following probability matrix,
\bqa\left[      \begin{array}{cc}
       \theta(1-a_1)^2+(1-\theta)a_1^2& a_1(1-a_1)\\
       a_1(1-a_1) &  \theta a_1^2+(1-\theta)(1-a_1)^2
      \end{array}
     \right].\label{distribution}\eqa
It has been shown by Witsenhausen \cite{Witsenhausen_CI76} that the common information of $X_1,X_2$
is achieved with $W$ being $S$. That is \bqa C(X_1,X_2)= I(X_1X_2;S)=H(X_1,X_2)-2h(a_1),\eqa
where $h(\cdot)$ is the binary entropy function.
When $\theta=\frac{1}{2}$, $(X_1,X_2)$  is a Doubly Symmetric Binary Source (DSBS) whose common information was derived by  Wyner \cite{Wyner_CI_75IT} using a different approach.

We now obtain the common information for $N$ variables.
\begin{proposition}
Let $S\sim$ Bern($\theta$) and let $X_i$, $i=1,\cdots, N$, be the output of independent BSCs with common input $S$ and
crossover probability $0\leq a_1\leq1/2$. Then for any $N\geq 2$, the common information for $X_1,\cdots,X_N$
 is given as
\bqa C(X_1,\cdots, X_N)=I(X_1,\cdots,X_N;S).\eqa
\end{proposition}
\begin{proof}
That $C(X_1,\cdots, X_N)\leq I(X_1,\cdots,X_N;S)$ follows from the definition of the common information (\ref{eq CI newdef}). The  inequality  $C(X_1,\cdots, X_N)\geq I(X_1,\cdots,X_N;S)$ can  be proved by contradiction.   Suppose there exists a $W$ such that
\bqa C(X_1,\cdots, X_N)=I(X_1,\cdots,X_N;W)<I(X_1,\cdots,X_N;S),\label{eq 33}\eqa i.e., $C(X_1,\cdots,X_N)$ is achieved by $W$ and it is strictly less than $I(X_1,\cdots,X_N;S)$. Since $W$ induces conditional independence of $X_1,\cdots, X_N$, we have, from (\ref{eq 33}),
\bqn \sum_{i=1}^NH(X_i|W)>\sum_{i=1}^NH(X_i|S).\eqn
Thus, there must exist two random variables $X_{k}, X_{j}$, $k,j\in\{1,\cdots, N\}$ such that
\bqn H(X_k|W)+H(X_j|W)>H(X_k|S)+H(X_j|S).\eqn
Given that the sequence $\{X_1,\cdots, X_N\}$ is exchangeable \cite{Bernardo_96}, $p(x_k,x_j)$ has the same joint distribution as $p(x_1,x_2)$. Thus,
\[C(X_1,X_2)=C(X_k,X_j)=I(X_k,X_j;W)<I(X_k,X_j;S)=I(X_1,X_2;S).\]
This, however, contradicts the fact that $S$ achieves $C(X_1,X_2)$. Thus the proposition is proved.
\end{proof}

We now characterize the minimum common rate $C_3(D_1,D_2)$ for a DSBS.

\begin{proposition}\label{proposition binary}
Consider a DSBS $(X_1,X_2)$ with distribution
\bqa
p(x_1,x_2)=\left\{\begin{array}{ll}
\frac{1}{2}(1-a_0), &\mbox{if $x_1=x_2$},\\
\frac{1}{2}a_0, &\mbox {otherwise},\\
\end{array}\right.\label{eq 32}
\eqa
where, without loss of generality, $0\leq a_0\leq 1/2$. Let $a_1$ be such that $a_0=2a_1(1-a_1), 0\leq a_1\leq1/2$.  With Hamming distortion $d_1=d_2=d_H$, we have
\bqa
C_3(D_1,D_2)=\left\{\begin{array}{ll}
C(X_1,X_2),& (D_1,D_2)\in \Emat_{10},\\
R_{X_1X_2}(D_1,D_2),& (D_1,D_2)\in \Emat_2\cup\Emat_3,\\
0,& (D_1,D_2)\geq (\frac{1}{2},\frac{1}{2}),
\end{array}\right.\eqa
\bqa C(X_1,X_2)\leq C_3(D_1,D_2)\leq R_{X_1X_2}(D_1,D_2), ~~(D_1,D_2)\in \Emat_{11},\eqa
        where
       \bqa  \begin{array}{lll}\Emat_{10}&=&\{(D_1,D_2):0\leq D_i\leq a_1, i=1,2\},\\
       \Emat_{11}&=&\Emat_{10}^c\cap\{(D_1,D_2):D_1+D_2-2D_1D_2\leq a_0\},\\
       \Emat_2 &=&\Emat_{10}^c\cap\Emat_{11}^c\cap\left\{(D_1,D_2):  \max\left\{\frac{D_1-D_2}{1-2D_2},\frac{D_2-D_1}{1-2D_1}\right\}\leq a_0\right\},\\
       \Emat_3&=&\Emat_{10}^c\cap\Emat_{11}^c\cap\Emat_2^c\cap\left\{(D_1,D_2): D_i\leq \frac{1}{2},i=1,2\right\}.\end{array}\label{binary region}\eqa
\end{proposition}
\begin{figure}
\centerline{
\begin{psfrags}
\psfrag{d1}[c]{$D_1$}
\psfrag{d2}[c]{$D_2$}
\psfrag{a1}[c]{$a_1$}
\psfrag{a0}[c]{$a_0$}
\psfrag{0.5}[c]{$\frac{1}{2}$}
\psfrag{0}[c]{$0$}
\psfrag{e10}[c]{$\Emat_{10}$}
\psfrag{e11}[c]{$\Emat_{11}$}
\psfrag{e2}[c]{$\Emat_2$}
\psfrag{e3}[c]{$\Emat_3$}
 \scalefig{.4}\epsfbox{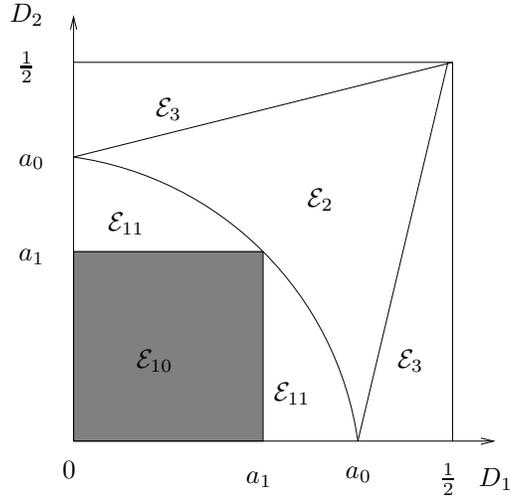}
\end{psfrags}}
\caption{\label{fig: Binary rate distortion region} The distortion regions $\Emat_{10},\Emat_{11},\Emat_2$ and $\Emat_3$  for the DSBS.
$C_3(D_1,D_2)=C(X_1,X_2)$ in the shaded region.}
\end{figure}
\vspace{0.1in}

\begin{proof}
  For $X_i\sim \mbox{Bern}(1/2),i=1,2$ with Hamming distortion, the rate distortion function is
      \bqn R_{X_i}(D_i)=\left\{\begin{array}{ll}1-h(D_i),& 0\leq D_i\leq \frac{1}{2},\\
      0,& D_i\geq \frac{1}{2}.\end{array} \right.\eqn
The joint rate distortion function
       of the DSBS $(X_1,X_2)$  is given by  \cite{Nayak_10IT}
\bqa R_{X_1X_2}(D_1,D_2)\!\!=\!\!\left\{\begin{array}{ll}
1+h(a_0)-h(D_1)-h(D_2),&(D_1,D_2)\in \Emat_1,\\
1-(1-a_0)h\left(\frac{D_1+D_2-a_0}{2(1-a_0)}\right)-a_0h\left(\frac{D_1-D_2+a_0}{2a_0}\right),&(D_1,D_2)\in \Emat_2,\\
1-h\left(\min\{D_1,D_2\}\right),&(D_1,D_2)\in \Emat_3.
\end{array} \right.\label{eq:lossybinary}\eqa
where $\Emat_1=\Emat_{10}\cup\Emat_{11}$ with $\Emat_{10},\Emat_{11},\Emat_2$ and $\Emat_3$ defined in (\ref{binary region}).
 Therefore, for this DSBS,
 $R_{X_1}(D_1)+R_{X_2}(D_2)-I(X_1;X_2)=R_{X_1X_2}(D_1,D_2)$, for $(D_1,D_2)\in \Emat_1$.
From Lemma \ref{lemma 4-0}, we have for $(D_1,D_2)\in \Emat_1$,
 \bqa C_3(D_1,D_2)\geq C(X_1,X_2). \label{eq 45}\eqa

 On the other hand,  the conditional rate distortion function $R_{X_i|S}(D_i)$, $i=1,2$, is given by \cite{Gray_73}
  \bqn R_{X_i|S}(D_i)
  =\left\{\begin{array}{ll}h(a_1)-h(D_i),& 0\leq D_i\leq a_1,\\
      0,& D_i\geq a_1.\end{array} \right.\eqn

Therefore, $ R_{X_1|S}(D_1)+R_{X_2|S}(D_2)+I(X_1,X_2;S)=R_{X_1X_2}(D_1,D_2)$ is satisfied for $(D_1,D_2)\in\Emat_{10}$.
By Theorem \ref{th 3}, $C_3(D_1,D_2)\leq C(X_1,X_2)$ for  $(D_1,D_2)\in\Emat_{10}$.
Together with (\ref{eq 45}) and given that $\Emat_{10}\subset\Emat_1$,
we have proved that for $(D_1,D_2)\in\Emat_{10}$,
\bqn
C_3(D_1,D_2)= C(X_1,X_2).\eqn

For $(D_1,D_2)\in \Emat_2$, we only need to show that $C_3(D_1,D_2)\geq R_{X_1X_2}(D_1,D_2)$. It was shown in \cite{Nayak_10IT} that 
the backward test channel that achieves $R_{X_1X_2}(D_1,D_2)$  is given by
\bqn  \begin{array}{cc}
X_1&=\hat{X}_1+Z_1,\\ X_2&=\hat{X}_2+Z_2,\end{array}\eqn
where both $\hat{X}_1,\hat{X}_2$ and $Z_1,Z_2$ are binary vectors independent of each other with the probability mass functions  given respectively as
\bqn P_{\hat{X}_1\hat{X}_2}=
\left[
      \begin{array}{cc}
       \frac{1}{2}& 0 \\
       0 & \frac{1}{2} \\
      \end{array}
     \right],\quad \quad
     P_{Z_1Z_2}=
\frac{1}{2}\left[
      \begin{array}{cc}
       2-a_0-D_1-D_2& D_2-D_1+a_0 \\
       D_1-D_2+a_0 & D_1+D_2-a_0 \\
      \end{array}
     \right].\eqn
 Therefore, $(\hat{X}_1,\hat{X}_2)$ that achieves $R_{X_1X_2}(D_1,D_2)$ satisfies \bqn \hat{X_2}=\hat{X}_1.\eqn
For the characterization   $C^*(D_1,D_2)$ of  $C_3(D_1,D_2)$, any  $W$ satisfying the Markov chain $\hat{X}_1-W-\hat{X}_1$ must satisfy $H(\hat{X}_1|W)=0$. Thus, $\hat{X}_1$ is a function of $W$ and we  have
\bqn I(X_1,X_2;W)=I(X_1,X_2;W,\hat{X}_1)\geq I(X_1,X_2;\hat{X}_1)=R_{X_1X_2}(D_1,D_2).\eqn
 Therefore,  $C_3(D_1,D_2)=R_{X_1X_2}(D_1,D_2)$.

The region $\Emat_3$ is a degenerated one. For example,  $R_{X_1X_2}(D_1,D_2)=R_{X_1}(D_1)$ if $a_0<\frac{D_2-D_1}{1-2D_1}$ and $ D_i\leq \frac{1}{2},i=1,2$. This  implies that the optimal coding scheme is to ignore $X_2$ and optimally compress $X_1$. Then $\hat{X}_2$ can be estimated from $\hat{X}_1$ with distortion less than $D_2$. The case of $a_0<\frac{D_1-D_2}{1-2D_2}$ is dealt with similarly. Hence, similar to the region $\Emat_2$, $C_3(D_1,D_2)=R_{X_1X_2}(D_1,D_2)$.
\end{proof}


The characterization of $C_3(D_1,D_2)$ is plotted in Fig.~\ref{fig: Binary rate distortion region} as a function of the distortion constraints. $C_3(D_1,D_2)=C(X_1,X_2)$ in the shaded region.  For the symmetric distortion constraint, $D_1=D_2=D$, the relation of $C_3(D,D)$ and $D$ for the DSBS is given in Fig.~\ref{fig: Binary symmetric figure}.
\begin{figure}
\centerline{
\begin{psfrags}
\psfrag{a}[c]{\hspace{-2.5cm}$C_3(D,D)$}
\psfrag{b}[c]{$D$}
\psfrag{c}[c]{$a_1$}
\psfrag{d}[c]{$\frac{1}{2}$}
\psfrag{0.5}[c]{$\frac{1}{2}$}
\psfrag{e}[c]{$0$}
\psfrag{f}[c]{$C(X_1,X_2)$}
\psfrag{g}[c]{$R_{X_1X_2}(D,D)$}
 \scalefig{.4}\epsfbox{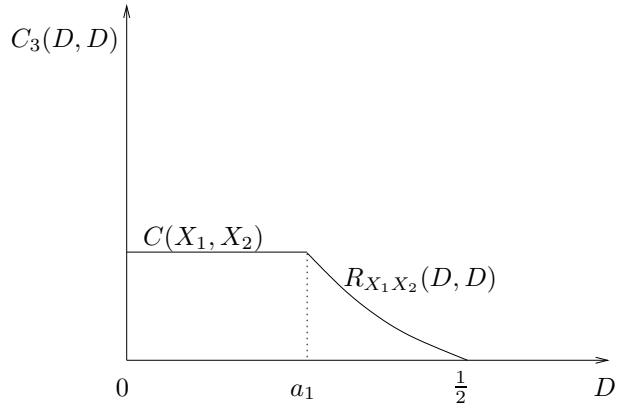}
\end{psfrags}}
\caption{\label{fig: Binary symmetric figure} The relation of $C_3(D,D)$ and $D$ for the DSBS with  $D_1=D_2=D$.}
\end{figure}


{\em Remarks:}
\bi
\item The claim $C_3(D_1,D_2)=C(X_1,X_2)$ for $(D_1,D_2)\in\Emat_{10}$ can also be proved using Theorem \ref{theorem 9}. $R_{X_1X_2}(a_1,a_1)$ is achieved by the backward test channel $p_b(x_1,x_2|s)=p(x_1|s)p(x_2|s)$. The vector  source $(X_1,X_2)$  is successively refinable for any $(D_1,D_2)\leq (a_1, a_1)$ \cite{Nayak_10IT} and the scalar source $X_i$ is successively refinable for any
$D_i\leq a_1$, $i=1,2$\cite{Equitz_91IT}. Thus by Theorem \ref{theorem 9}, $C_3(D_1,D_2)=C(X_1,X_2)$ for $(D_1,D_2)\leq (a_1, a_1)$.
\item We have the full characterization of $C_3(D_1,D_2)$ in the distortion region except the region $\Emat_{11}$. From the proof of Proposition \ref{proposition binary}, we know that $C_3(D_1,D_2)\geq C(X_1,X_2)$ for $(D_1,D_2)\in\Emat_{11}$, but the exact value of $C_3(D_1,D_2)$ in this region remains unknown.
\item
Let $(D_1,D_2)\leq (D_1',D_2')\leq (a_1,a_1)$, then the rate $R_{X_1X_2}(D_1',D_2')$ is $(D_1,D_2)-$achievable in the Gray-Wyner network, i.e., $R_{X_1X_2}(D_1',D_2')\geq C_3(D_1,D_2)$.

To show this, let $(\hat{X}_1,\hat{X}_2)$ achieve $R_{X_1X_2}(D_1',D_2')$.
The backward test channel that achieves $R_{X_1X_2}(D_1',D_2')$  satisfies  $p_b(x_1,x_2|\hat{x}_1\hat{x}_2)=p_b(x_1|\hat{x}_1)p_b(x_2|\hat{x}_2)$ where
\bqn p_b(x_i|\hat{x}_i)=\left\{\begin{array}{lr}
1-D_i', &\mbox{ if $x_i=\hat{x}_i$},\\
D_i', &\mbox {Otherwise}.\\
\end{array}\right.\eqn
for $i=1,2$. 
Then for $(D_1,D_2)\leq (D_1',D_2')\leq (a_1,a_1)$, let the rate allocation of $R_0,R_1,R_2$ in the Gray-Wyner network be
\bqa\begin{array}{l}R_0=R_{X_1X_2}(D_1',D_2')=1+h(a_0)-h(D_1')-h(D_2'),\\
R_i=R_{X_i|\hat{X}_1\hat{X}_2}(D_i)=R_{X_i|\hat{X}_i}(D_i)=h(D_i')-h(D_i),i=1,2.\end{array}\label{eq binary}\eqa
Since $R_0,R_1$ and $R_2$ in (\ref{eq binary}) sum up to $R_{X_1X_2}(D_1,D_2)$, $R_{X_1X_2}(D_1',D_2')$ is $(D_1,D_2)-$achievable.

  The minimal $R_0$  satisfying (\ref{eq binary}) is exactly $C(X_1,X_2)$, which is achieved by  letting $(D_1',D_2')= (a_1,a_1)$.
\ei

\subsection{Gaussian random variables}
In this section we consider bivariate Gaussian random variables $X_1,X_2$ with zero mean and covariance matrix
 \bqa K_2=\left[
      \begin{array}{cc}
       \sigma_1^2& \rho \sigma_1\sigma_2\\
       \rho\sigma_1\sigma_2& \sigma_2^2
      \end{array}
     \right].\label{matrix 1}\eqa
The common information between this pair of Gaussian random variables is given in the following theorem.
\begin{theorem}\label{th Gaussian}
 For  two joint Gaussian random variables $X_1,X_2$ with covariance matrix $K_2$,
     the common information is
    \bqa C(X_1,X_2)=\frac{1}{2}\log\frac{1+\rho}{1-\rho}.\eqa
\end{theorem}
\vspace{0.1in}
\begin{proof} See Appendix \ref{appendix 6}.\end{proof}

 As the common information of $(X_1,X_2)$ is only a function of the correlation coefficient  $\rho$,  we  consider, without loss of generality, the covariance matrix  \bqa K'_2= \left[
      \begin{array}{cc}
       1& \rho \\
       \rho& 1
      \end{array}
     \right].\label{matrix K_2'}\eqa

The above result generalizes to multi-variate Gaussian random variables satisfying a certain covariance matrix structure, the proof of which can be constructed in a similar fashion.
\begin{corollary}
For  $N$ joint Gaussian random variables $X_1,X_2,\cdots,X_N$ with covariance matrix $K_N$,
 \bqa K_N=\left[
      \begin{array}{cccc}
       1& \rho &\cdots &\rho\\
       \rho & 1 &\cdots&\rho\\
        \cdot &\cdot&\cdots&\cdot\\
         \rho&\rho &\cdots&1
      \end{array}
     \right],\label{matrix 1}\eqa
     the common information is
    \bqa C(\Xbf_N)=\frac{1}{2}\log\left(1+\frac{N\rho}{1-\rho}\right).\eqa
\end{corollary}
\vspace{0.1in}

 We now characterize the minimum common rate $C_3(D_1,D_2)$ in the Gray-Wyner lossy source coding network for bivariate Gaussian random variables with covariance matrix $K_2'$ in equation (\ref{matrix K_2'}).
 It was shown in \cite{Viswanatha&Akyol&Rose:ISIT12} that for symmetric distortion, i.e.,$D_1=D_2=D$,
 \bqa
C_3(D,D)=\left\{\begin{array}{ll}
C(X_1,X_2),& 0\leq D\leq 1-\rho,\\
R_{X_1X_2}(D,D),&  1-\rho\leq D\leq 1,\\
0,& D\geq 1.
\end{array}\right.\eqa
We characterize $C_3(D_1,D_2)$ for general distortion $(D_1,D_2)$ in the following proposition.
\begin{proposition}
For bivariate Gaussian random variables $X_1,X_2$ with zero mean, covariance matrix $K'_2$ and  squared error distortion, we have
that \bqa
C_3(D_1,D_2)=\left\{\begin{array}{ll}
C(X_1,X_2),& (D_1,D_2)\in\Dmat_{10},\\
R_{X_1X_2}(D_1,D_2),& (D_1,D_2)\in \Dmat_2\cup\Dmat_3,\\
0,&(D_1,D_2)\geq (1,1),
\end{array}\right.\eqa
\bqa C(X_1,X_2)\leq C_3(D_1,D_2)\leq R_{X_1X_2}(D_1,D_2),~~(D_1,D_2)\in \Dmat_{11},\eqa
where
\bqa\begin{array}{lll}
\Dmat_{10}&=&\{(D_1,D_2): 0\leq D_i\leq 1-\rho,i=1,2\},\\
\Dmat_{11}&=&\Dmat_{10}^c\cap\{(D_1,D_2):D_1+D_2-D_1D_2\leq 1-\rho^2\},\\
\Dmat_2&=&\Dmat_{10}^c\cap\Dmat_{11}^c\cap\left\{(D_1,D_2): \min\left\{\frac{1-D_1}{1-D_2},\frac{1-D_2}{1-D_1}\right\}\geq \rho^2\right\},\\
\Dmat_3&=&\Dmat_{10}^c\cap\Dmat_{11}^c\cap\Dmat_2^c\cap \{(D_1,D_2): D_i\leq 1,i=1,2\}.\end{array}\eqa
\end{proposition}
\vspace{0.1in}
\begin{proof}
The joint rate distortion function for Gaussian random variables with  squared error distortion\cite{Berger_RDT,Xiao_Luo_05,Nayak_10IT} is given by
\bqa  R_{X_1X_2}(D_1,D_2)=\left\{\begin{array}{ll}
\frac{1}{2}\log\frac{1-\rho^2}{D_1D_2}, &(D_1,D_2)\in \Dmat_1,\\
\frac{1}{2}\log\frac{1-\rho^2}{D_1D_2-\left(\rho-\sqrt{(1-D_1)(1-D_2)}\right)^2},&(D_1,D_2)\in \Dmat_2,\\
\frac{1}{2}\log\frac{1}{\min\{D_1,D_2\}}, &(D_1,D_2)\in\Dmat_3,
\end{array}\right.\eqa
where $\Dmat_1=\Dmat_{10}\cup\Dmat_{11}$.
The marginal rate distortion function for $X_i\sim \Nmat(0,1),i=1,2$, is
\bqn R_{X_i}(D_i)=\left\{\begin{array}{ll}
\frac{1}{2}\log\frac{1}{D_i},&0\leq D_i\leq 1,\\
0,& D_i\geq 1.
\end{array}\right.\eqn

Therefore,
$R_{X_1}(D_1)+R_{X_2}(D_2)-I(X_1;X_2)=R_{X_1X_2}(D_1,D_2)$, for $(D_1,D_2)\in \Dmat_1$.
From Lemma \ref{lemma 4-0}, for $(D_1,D_2)\in \Dmat_1$,
\bqn C_3(D_1,D_2)\geq C(X_1,X_2).\eqn

On the other hand,  the random variable $W$ in the following decomposition of $X_1$ and $X_2$ achieves the common information
\bqa X_i=\sqrt{\rho}W+\sqrt{1-\rho}N_i,~ i=1,2.\eqa
where $W,N_1,N_2$ are mutually independent standard Gaussian random variables.
The conditional distribution of $X$ given $W$ is Gaussian distribution with variance $1-\rho$.
Hence, for $i=1,2$, the conditional rate distortion function is
\bqa R_{X_i|W}(D_i)=\left\{\begin{array}{ll}\frac{1}{2}\log\frac{1-\rho}{D_i},&0\leq D_i\leq 1-\rho, \\
0,& D_i\geq 1-\rho.
\end{array}\right.
\eqa
The condition $R_{X_1|W}(D_1)+R_{X_2|W}(D_2)+I(X_1,X_2;W)=R_{X_1X_2}(D_1,D_2)$ is satisfied for $(D_1,D_2)\in\Dmat_{10}$.
From Theorem \ref{th 3},
 $C_3(D_1,D_2)\leq C(X_1,X_2)$ for $(D_1,D_2)\in\Dmat_{10}$.
Since, $\Dmat_{10}\in\Dmat_1$, we proved that for $(D_1,D_2)\in\Dmat_{10}$,\bqn
C_3(D_1,D_2)= C(X_1,X_2).\eqn


For $(D_1,D_2)\in \Dmat_2$, it was shown in \cite{Nayak_10IT} that $(\hat{X}_1,\hat{X}_2)$ that achieves
$R_{X_1X_2}(D_1,D_2)$ satisfies \bqn \hat{X_2}=\sqrt{\frac{1-D_2}{1-D_1}}\hat{X}_1.\eqn
Hence, using  the characterization   $C^*(D_1,D_2)$, it is easy to show that the $W$ satisfying the Markov chains (\ref{eq 1}) and (\ref{eq 2}) must satisfy  two Markov chains
 \bqn X_1X_2-\hat{X}_1-W-\hat{X}_2,\\
 X_1X_2-\hat{X}_2-W-\hat{X}_1.\eqn
 Therefore, we have \bqn I(X_1,X_2;W)=I(X_1,X_2;\hat{X}_1)=I(X_1,X_2;\hat{X}_1,\hat{X}_2),\eqn which proved $C_3(D_1,D_2)=R_{X_1X_2}(D_1,D_2)$.

The region $\Dmat_3$ is a degenerated one. For example,  $R_{X_1X_2}(D_1,D_2)=R_{X_1}(D_1)$ if $\frac{1-D_2}{1-D_1}<\rho^2$, this means that the correlation between $X_1$ and $X_2$ is so strong that the optimal coding scheme is to  encode $X_1$ to within distortion $D_1$ and ignore $X_2$. Then $\hat{X}_2$ can be estimated from $\hat{X}_1$. We have
 \bqn \hat{X}_2=\rho\hat{X}_1.\eqn
The case of $\frac{1-D_1}{1-D_2}<\rho^2$ is dealt with similarly. Hence, we have $C_3(D_1,D_2)=R_{X_1X_2}(D_1,D_2)$.
\end{proof}
\begin{figure}
\centerline{
\begin{psfrags}
\psfrag{d1}[c]{$D_1$}
\psfrag{d2}[c]{$D_2$}
\psfrag{a1}[c]{$1-\rho$}
\psfrag{a0}[c]{$1-\rho^2$}
\psfrag{0.5}[c]{$1$}
\psfrag{0}[c]{$0$}
\psfrag{e10}[c]{$\Dmat_{10}$}
\psfrag{e11}[c]{$\Dmat_{11}$}
\psfrag{e2}[c]{$\Dmat_2$}
\psfrag{e3}[c]{$\Dmat_3$}
 \scalefig{.5}\epsfbox{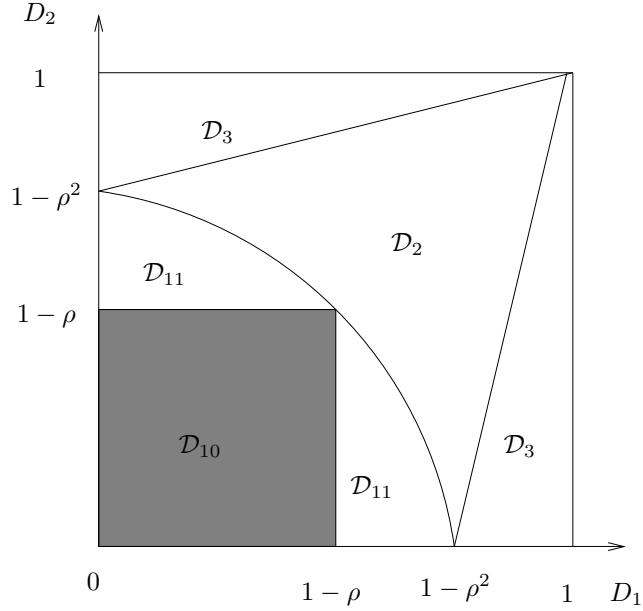}
\end{psfrags}}
\caption{\label{fig: Gaussian distortion region} The distortion regions $\Dmat_{10}, \Dmat_{11},\Dmat_2$ and $\Dmat_3$  for bivariate Gaussian random variables.
$C_3(D_1,D_2)=C(X_1,X_2)$ in the shaded region.}
\end{figure}

The characterization of $C_3(D_1,D_2)$ is plotted in Fig. \ref{fig: Gaussian distortion region} as a function of the distortion constraints. $C_3(D_1,D_2)=C(X_1,X_2)$ in the shaded region.

{\em Remarks:}
\bi
\item Similar to the binary case, the claim $C_3(D_1,D_2)=C(X_1,X_2)$ for $(D_1,D_2)\in\Dmat_{10}$
can also be proved using Theorem \ref{theorem 9}. This is because for the bivariate Gaussian random variables with covariance matrix $K'_2$, $R_{X_1X_2}(1-\rho,1-\rho)$ is achieved by the backward test channel $p_b(x_1,x_2|w)=p(x_1|w)p(x_2|w)$, $(X_1,X_2)$  is successively refinable for any $(D_1,D_2)\leq (1-\rho, 1-\rho)$  \cite{Nayak_10IT} and $X_i$ is successively refinable for $D_i\leq 1-\rho$, $i=1,2$\cite{Equitz_91IT}.
\item Similarly,  $C_3(D_1,D_2)\geq C(X_1,X_2)$ for $(D_1,D_2)\in\Dmat_{11}$ but the exact characterization of $C_3(D_1,D_2)$ remains unknown in this region.
\item
Let $(D_1,D_2)\leq (D_1',D_2')\leq (1-\rho,1-\rho)$, then the rate $R_{X_1X_2}(D_1',D_2')$ is $(D_1,D_2)-$achievable in the Gray-Wyner network, i.e., $R_{X_1X_2}(D_1',D_2')\geq C_3(D_1,D_2)$.

This is because for $(D_1',D_2')\in\Emat_{10}$, the joint rate distortion function $R_{X_1X_2}(D_1',D_2')$  is achieved by Gaussian distributed  $(\hat{X}_1,\hat{X}_2)$  satisfying $X_1-\hat{X}_1-\hat{X}_2-X_2$ where the covariance matrix of $(\hat{X}_1,\hat{X}_2)$  is \cite{Nayak_10IT}\bqn K_{\hat{X}_1\hat{X}_2}= \left[
      \begin{array}{cc}
       1-D_1'& \rho \\
       \rho& 1-D_2'
      \end{array}
     \right].\eqn
Then for $(D_1,D_2)\leq (D_1',D_2')\leq (1-\rho,1-\rho)$, let the rate allocation of $R_0,R_1,R_2$ for the Gray-Wyner network be as follows:
\bqa\begin{array}{l}R_0=R_{X_1X_2}(D_1',D_2')=\frac{1}{2}\log\frac{1-\rho^2}{D_1'D_2'},\\
R_i=R_{X_i|\hat{X}_1\hat{X}_2}(D_i)=R_{X_i|\hat{X}_i}(D_i)=\frac{1}{2}\log\frac{D_i'}{D_i},i=1,2.\end{array}\label{eq Gaussisan}\eqa
 $R_0,R_1$ and $R_2$ in (\ref{eq Gaussisan}) sum up to $R_{X_1X_2}(D_1,D_2)$, so $R_{X_1X_2}(D_1',D_2')$ is $(D_1,D_2)-$achievable.

 Therefore,
in the Gray-Wyner network, we can use the rate allocation in (\ref{eq Gaussisan}) to achieve the distortion $(D_1,D_2)\leq (1-\rho, 1-\rho)$ for any $(D_1,D_2)\leq (D_1',D_2')\leq (1-\rho,1-\rho)$. The minimal $R_0$ satisfying (\ref{eq Gaussisan}) is exactly $C(X_1,X_2)$, which is achieved by  letting $(D_1',D_2')= (1-\rho,1-\rho)$.
\ei

\section{Conclusion}\label{section 5}
We have generalized the definition of Wnyer's common information and expanded its practical significance by providing a new operational interpretation. The generalization is two-folded: the number of dependent variables can be arbitrary, so are the alphabet of those random variables. We have determined  new properties  for the generalized Wyner's common information of $N$ dependent variables. More importantly, we have derived a lossy source coding interpretation of Wyner's common information using the Gray-Wyner network. In particular, it is established that the common information is precisely  the smallest common message rate when the total rate is arbitrarily close to the rate distortion function with joint decoding. A surprising observation is that such equality holds independent of the values of distortion constraints as long as the distortions are within some distortion region.
Two examples, the doubly symmetric binary source under Hamming distortion and bivariate Gaussian source under square-error distortion, are used to illustrate the lossy source coding interpretation of Wyner's common information. The common information for bivariate Gaussian source and its extension to the multi-variate case has also been computed explicitly.

While the lossy source coding interpretation of Wyner's common information presented in this paper is limited to $N=2$ random variables, the results can be extended to arbitrary $N$ random variables.

\section*{Acknowledgment}
The authors gratefully acknowledges Professor Paul Cuff of Princeton University for pointing out an error in an earlier proof of Theorem \ref{th Gaussian} given in \cite{Xu_11_CISS}.

\appendix

\subsection{Proof of Theorem \ref{th 3}}\label{appendix 3}

 We first  show that $C_3(D_1,D_2)\geq \tilde{C}(D_1,D_2)$. Let $R_0$ be $(D_1,D_2)$-achievable, then there exists an $(n, M_0, M_1, M_2)$
code  such that (\ref{CI EQ1})-(\ref{CI EQ3}) are satisfied. Define $R_i=\frac{1}{n}\log M_i$ for $i=1,2$.
Since $(R_0,R_1,R_2)$ is $(D_1,D_2)$-achievable, from Theorem \ref{them A2},  there exists a $W$ such that
\bqn R_0&\geq& I(X_1,X_2; W),\\
R_i&\geq& R_{X_i|W}(D_i),\hspace{0.2in}i=1,2\eqn
and  for any $\epsilon>0$,
 \beq\sum_{i=0}^2R_i\leq R_{X_1X_2}(D_1,D_2)+\epsilon.\label{eq 39}\eeq
Therefore,
\bqa R_{X_1X_2}(D_1,D_2)+\epsilon
&\geq&\sum _{i=0}^2R_i\nn\\
&\geq&\!\!\! \!\!I(X_1,X_2; W)\!+\!\sum_{i=1}^2R_{X_i|W}(D_i)\nn\\
                \!\!\!\!&\geq&\!\! \!\!I(X_1,X_2; W)\!+\!R_{X_1X_2|W}(D_1,D_2)\label{eq 21}\\
                \!\!\!\!&\geq&\!\!\!\! R_{X_1X_2}(D_1,D_2)\label{eq 25}\eqa
where (\ref{eq 21}) is from  (\ref{Rd 2}) and (\ref{eq 25}) comes from (\ref{Rd 4}).
Thus, we have  \bqa
I(X_1,X_2; W)\!+\!R_{X_1|W}(D_1)+R_{X_2|W}(D_2)=R_{X_1X_2}(D_1,D_2)\label{eq 22}.
\eqa

Hence, if $R_0$ is $(D_1,D_2)$-achievable, there exists a $W$ such that
$R_0\geq I(X_1,X_2;W)$ and (\ref{eq 22}) is true. It shows that
 $C_3(D_1,D_2)\geq \tilde{C}(D_1,D_2)$.

Next we show  $C_3(D_1,D_2)\leq \tilde{C}(D_1,D_2)$. Let $W'$ be the random variable that achieves $\tilde{C}(D_1,D_2)$.
For any $R_0> \tilde{C}(D_1,D_2)$ and $\epsilon>0$,  let
\beq \epsilon_1=\min \left\{\frac{\epsilon}{3}, R_0-\tilde{C}(D_1,D_2)\right\},\label{eq 27}\eeq
and hence $\epsilon_1>0$. From theorem \ref{them A2}, there exists an $(n,M_0,M_1,M_2)$
code with $Ed_1(X_1,\hat{X}_1)\leq D_1$, $Ed_2(X_2,\hat{X}_2)\leq D_2$,  and
\bqa \frac{1}{n}\log M_0&\leq& I(X_1,X_2;W')+\epsilon_1
=\tilde{C}(D_1,D_2)+\epsilon_1\leq R_0,\label{eq 17} \\
\frac{1}{n}\log M_i&\leq& R_{X_i|W'}(D_i)+\epsilon_1,\label{eq 18}\eqa
for $i=1,2$.
Sum over (\ref{eq 17}) and (\ref{eq 18}), we get
\bqa\sum^2_{i=0}\frac{1}{n}\log{M_i}
&\leq&I(X_1,X_2;W')+\sum^2_{i=1}R_{X_i|W'}(D_i)+3\epsilon_1\nn\\
&\leq & R_{X_1X_2}(D_1,D_2)+\epsilon\label{eq 38},
\eqa
where inequality (\ref{eq 38}) comes from  (\ref{eq 27}) and definition of $\tilde{C}(D_1,D_2)$.

This proves that $R_0$ is
$(D_1,D_2)$-achievable, thus completes the proof of $C_3(D_1,D_2)\leq \tilde{C}(D_1,D_2)$.

\subsection{Direct proof of  $\tilde{C}(D_1,D_2)= C^*(D_1,D_2)$}\label{appendix 3-0}

First we show that $\tilde{C}(D_1,D_2)\geq C^*(D_1,D_2)$.

Let $W$ be the variable that achieves $\tilde{C}(D_1,D_2)$ and let $\hat{X}_1,\hat{X}_2$ be random variables that achieve $R_{X_1|W}(D_1)$ and $R_{X_2|W}(D_2)$, i.e.,
\bqa  I(X_1,X_2;W)+R_{X_1|W}(D_1)+R_{X_2|W}(D_2)&=&R_{X_1X_2}(D_1,D_2),\label{eq 11}\\
R_{X_1|W}(D_1)&=&I(X_1;\hat{X}_1|W),\\
R_{X_2|W}(D_2)&=&I(X_2;\hat{X}_2|W),\\
 E[d_1(X_1,\hat{X}_1)]&\leq& D_1,\\
 E[d_2(X_2,\hat{X}_2)]&\leq& D_2.\label{eq 29}
\eqa
Without loss of generality, we can assume that the joint distribution of $(X_1,X_2,\hat{X}_1,\hat{X}_2,W)$ factors as
 $p(x_1,x_2,\hat{x}_1,\hat{x}_2,w)=p(x_1,x_2,w)p(\hat{x}|x,w)p(\hat{y}|y,w)$ because the distortion $D_1$ is independent of $X_2$ and
$D_2$ is independent of $X_1$.
We now establish \bqn R_{X_1X_2|W}(D_1,D_2)=R_{X_1|W}(D_1)+R_{X_2|W}(D_2).\eqn This is  from (\ref{eq 11}) and the inequalities
\bqn
R_{X_1X_2|W}(D_1,D_2)+I(X_1,X_2;W)&\geq& R_{X_1X_2}(D_1,D_2),\\
R_{X_1|W}(D_1)+R_{X_2|W}(D_2)&\geq& R_{X_1X_2|W}(D_1,D_2),\eqn
 from Lemma \ref{lemma 1}. Therefore, together with (\ref{eq 11})-(\ref{eq 29}), we have
\bqa R_{X_1X_2|W}(D_1,D_2)
&=&I(X_1;\hat{X}_1|W)+I(X_2; \hat{X}_2|W)\nn\\
&=&H(\hat{X}_1|W)+H(\hat{X}_2|W)-H(\hat{X}_1|X_1,W)-H(\hat{X}_2|X_2,W)\nn\\
&\geq& H(\hat{X}_1,\hat{X}_2|W)-H(\hat{X}_1|X_1,W)-H(\hat{X}_2|X_2,W)\nn\\
&=& H(\hat{X}_1,\hat{X}_2|W)-H(\hat{X}_1|W,X_1,X_2)-H(\hat{X}_2|W,X_1,X_2)\nn\\
&= &I(X_1,X_2;\hat{X}_1,\hat{X}_2|W)\nn\\
&\geq &R_{X_1X_2|W}(D_1,D_2).\nn\eqa
As the left-hand side (LHS) and right-hand side (RHS) of the above inequalities are the same, all the inequalities must be equalities so we have
\bqn I(\hat{X}_1;\hat{X}_2|W)=0.\eqn
Then we have
\bqa R_{X_1X_2}(D_1,D_2)&=&I(X_1,X_2;W)+ I(X_1;\hat{X}_1|W)+I(X_2; \hat{X}_2|W)\nn\\
&=&I(X_1,X_2;W,\hat{X}_1,\hat{X}_2)-I(X_1,X_2;\hat{X}_1,\hat{X}_2|W)+I(X_1;\hat{X}_1|W)+I(X_2; \hat{X}_2|W)\nn\\
&=&I(X_1,X_2;\hat{X}_1,\hat{X}_2)+I(X_1,X_2;W|\hat{X}_1,\hat{X}_2)\nn\\
&\geq&I(X_1,X_2;\hat{X}_1,\hat{X}_2)\nn\\
&\geq&R_{X_1X_2}(D_1,D_2).\nn
\eqa


As the LHS and RHS of the above inequalities are the same, all the inequalities must be equalities so we have
\bqn I(X_1,X_2;W|\hat{X}_1,\hat{X}_2)&=&0,\\
I(X_1,X_2;\hat{X}_1,\hat{X}_2)&=&R_{X_1X_2}(D_1,D_2).\eqn
Therefore,  $X_1,X_2,\hat{X}_1,\hat{X}_2,W$ satisfy the Markov chains in (\ref{eq 1}) and (\ref{eq 2}) and $\hat{X}_1,\hat{X}_2$ achieve
$R_{X_1X_2}(D_1,D_2)$. Thus, $\tilde{C}(D_1,D_2)\geq C^*(D_1,D_2)$.

Next we  show that $\tilde{C}(D_1,D_2)\leq C^*(D_1,D_2)$.

Let $X_1,X_2,X^*_1,X^*_2,W$ achieve $C^*(D_1,D_2)$. Therefore, they satisfy the Markov chains in (\ref{eq 1}) and (\ref{eq 2}) and $I(X_1,X_2;X^*_1,X^*_2)=R_{X_1X_2}(D_1,D_2)$
and $E[d_1(X_1,X_1^*)]\leq D_1,E[d_2(X_2,X^*_2)]\leq D_2$.
\bqa R_{X_1X_2}(D_1,D_2)
&=&I(X_1,X_2;X^*_1,X^*_2)\nn\\
&=&I(X_1,X_2;W,X^*_1,X^*_2)\label{eq 4}\\
&=&I(X_1,X_2;W)+I(X_1,X_2;X^*_1,X^*_2|W)\nn\\
&=&I(X_1,X_2;W)+H(X^*_1|W)+H(X^*_2|W)-H(X^*_1,X^*_2|X_1,X_2,W)\label{eq 40}\\
&=&I(X_1,X_2;W)+I(X_1;X^*_1|W)+I(X_2;X^*_2|W)+H(X^*_1|X_1,W)\nn\\
&&+H(X^*_2|X_2,W)-H(X^*_1,X^*_2|X_1,X_2,W)\nn\\
&\geq&I(X_1,X_2;W)+I(X_1;X^*_1|W)+I(X_2;X^*_2|W)+H(X^*_1|X_1,X_2,W)\nn\\
&&+H(X^*_2|X_1,X_2,W)-H(X^*_1,X^*_2|X_1,X_2,W)\label{eq 5}\\
&=&I(X_1,X_2;W)+I(X_1;X^*_1|W)+I(X_2;X^*_2|W)+I(X^*_1;X^*_2|X_1,X_2,W)\nn\\
&\geq&I(X_1,X_2;W)+I(X_1;X^*_1|W)+I(X_2;X^*_2|W)\nn\\
&\geq& I(X_1,X_2;W)+R_{X_1|W}(D_1)+R_{X_2|W}(D_2)\nn\\
&\geq& I(X_1,X_2;W)+R_{X_1X_2|W}(D_1,D_2)\label{eq 7}\\
&\geq& R_{X_1X_2}(D_1,D_2),\label{eq 8}
\eqa
where (\ref{eq 4}) is from the Markov chain $(X_1,X_2)-(X^*_1,X^*_2)-W$, (\ref{eq 40}) is from the Markov chain
$X^*_1-W-X^*_2$, (\ref{eq 5}) is because conditioning reduces entropy, (\ref{eq 7}) and (\ref{eq 8})
 are by the properties of rate distortion functions.
As the LHS and RHS of the above inequalities are the same, all the inequalities must be equalities so we have
\bqn I(X_1,X_2;W)+R_{X_1|W}(D_1)+R_{X_2|W}(D_2)=R_{X_1X_2}(D_1,D_2).\eqn

Therefore, $C^*(D_1,D_2)=I(X_1,X_2;W)\geq \tilde{C}(D_1,D_2)$.

\subsection{Proof of Lemma \ref{lemma 4-0}}\label{Appendix 4-0}
Let $W$ be the random variable that achieves $C_3(D_1,D_2)$. Thus, $C_3(D_1,D_2)= I(X_1,X_2;W)$ with
\bqa R_{X_1|W}(D_1)+R_{X_2|W}(D_2)+I(X_1,X_2; W)=R_{X_1X_2}(D_1,D_2)\label{eq 14}.\eqa
Combined with (\ref{eq lemma 4-0}), we have that
\bqa R_{X_1}(D_1)+R_{X_2}(D_2)-I(X_1;X_2)
&=& R_{X_1|W}(D_1)+R_{X_2|W}(D_2)+I(X_1,X_2; W)\label{eq 15}\\
&\geq&R_{X_1}(D_1)-I(X_1;W)+R_{X_2}(D_2)-I(X_2;W)\nn\\
&&+I(X_1,X_2; W)\label{eq 16}\\
&=& R_{X_1}(D_1)+R_{X_2}(D_2)-I(X_1;X_2)+I(X_1;X_2|W)\label{eq 24}\\
&\geq& R_{X_1}(D_1)+R_{X_2}(D_2)-I(X_1;X_2),\label{eq 26}
 \eqa
 where equation (\ref{eq 15}) is from equations (\ref{eq 14}) and (\ref{eq lemma 4-0}), inequality (\ref{eq 16}) comes from Lemma \ref{lemma 1},
(\ref{eq 24}) is by the chain rule and inequality (\ref{eq 26}) is by the fact that $I(X_1;X_2|W)\geq 0$.

Because the LHS of (\ref{eq 15}) is the same as the RHS of (\ref{eq 26}), we  can conclude that all the inequalities above should be equalities. This implies  $I(X_1;X_2|W)=0$. Therefore,
\bqn C_3(D_1,D_2)\geq C(X_1,X_2).\eqn
\subsection{Proof of Theorem \ref{th 4}}\label{Appendix 4}

Let $W$ be the random variable that achieves the common information of $X_1,X_2$.
By Lemma \ref{lemma 2},  there exists a strictly positive surface $\Dmat(X_1X_2|W)$ such that  for any $0\leq (D_1, D_2)\leq \Dmat(X_1X_2|W)$,
 \bqa I(X_1,X_2;W)+R_{X_1X_2|W}(D_1,D_2)= R_{X_1X_2}(D_1,D_2)\label{eq 54}.\eqa

Also by Lemma \ref{lemma 2}, there exists a strictly positive surface $\Dmat(X_1X_2)\geq \Dmat(X_1X_2|W)$
 such that  for any $0\leq (D_1, D_2)\leq\Dmat(X_1X_2)$,
 \beq R_{X_1}(D_1)+R_{X_2}(D_2)-I(X_1;X_2)= R_{X_1X_2}(D_1,D_2).
\label{eq 9}\eeq
Since $\Dmat(X_1X_2|W)\leq \Dmat(X_1X_2)$, let $\gamma=\Dmat(X_1X_2|W)$,  both equalities (\ref{eq 54}) and  (\ref{eq 9}) hold  for $0\leq (D_1, D_2)\leq \gamma$. Therefore, from Lemmas \ref{lemma 3} and \ref{lemma 4-0},
$C_3(D_1,D_2)=C(X_1,X_2)$ for $0\leq (D_1, D_2)\leq \gamma$.

\subsection{Proof of Theorem \ref{theorem 9}}\label{Appendix them 9}

First we  show that for any $(D_1,D_2)\leq (D_1^0,D_2^0)$, 
\bqa R_{X_1X_2|W}(D_1,D_2)+I(X_1X_2;W)=R_{X_1X_2}(D_1,D_2)\label{eq 44}.\eqa

From the definition of $(D_1^0,D_2^0)$  in (\ref{ratedistortion}), we have
\bqn R_{X_1X_2}(D_1^0,D_2^0)\geq R_{X_1}(D_1^0)+R_{X_2}(D_2^0)-I(X_1;X_2)= I(X_1,X_2;W),\eqn
where the first inequality is from (\ref{Rd 5}).
On the other hand, \bqn R_{X_1X_2}(D_1^0,D_2^0)\leq I(X_1,X_2;\hat{X}_1^0,\hat{X}_2^0)\leq I(X_1,X_2;W).\eqn
Therefore, $R_{X_1X_2}(D_1^0,D_2^0)=I(X_1,X_2;\hat{X}_1^0,\hat{X}_2^0)=I(X_1X_2;W)$.

 Let $(\hat{X}_1,\hat{X}_2)$ achieve $R_{X_1X_2}(D_1,D_2)$.
As the vector source $(X_1,X_2)$ is successively refinable under individual distortion constraints\cite{Nayak_10IT}, we have the Markov chain $X_1X_2-\hat{X}_1\hat{X}_2-\hat{X}_1^0\hat{X}_2^0$. Therefore,
 \bqn R_{X_1X_2}(D_1,D_2)-I(X_1,X_2;W)
 &=& I(X_1,X_2;\hat{X}_1,\hat{X}_2)-I(X_1,X_2;\hat{X}_1^0,\hat{X}_2^0)\\
 &=&I(X_1X_2;\hat{X}_1\hat{X}_2|\hat{X}_1^0,\hat{X}_2^0)\\
 &\geq & R_{X_1X_2|\hat{X}_1^0,\hat{X}_2^0}(D_1,D_2)\\
  &\geq & R_{X_1X_2|W}(D_1,D_2),
 \eqn
where the last inequality is from the Markov chain $X_1X_2-W-\hat{X}_1^0,\hat{X}_2^0$.
 On the other hand, by Lemma \ref{lemma 1}, we have \bqn
 R_{X_1X_2|W}(D_1,D_2)+I(X_1X_2;W)\geq R_{X_1X_2}(D_1,D_2).\eqn
This establishes (\ref{eq 44}). Thus, from Lemma \ref{lemma 3}, $C_3(D_1,D_2)\leq C(X_1;X_2)$.

To complete the proof, we need to show
\bqa R_{X_1}(D_1)+R_{X_2}(D_2)-I(X_1;X_2)=R_{X_1X_2}(D_1,D_2)\label{eq 62}.\eqa
From Lemma \ref{lemma 1},
\bqn R_{X_1}(D_1)+R_{X_2}(D_2)-I(X_1;X_2)\leq R_{X_1X_2}(D_1,D_2).\eqn
 Therefore,  we only need to establish the other direction.
For $i=1,2$, let $\hat{X}_i$ achieve $R_{X_i}(D_i)$, then by the definition of a successively refinable  scalar source \cite{Equitz_91IT},
we have the Markov chain $X_i-\hat{X}_i-\hat{X}_i^0$ for $D_i\leq D_i^0$. Therefore,
 \bqa R_{X_i}(D_i)-I(X_i;W)&=&I(X_i;\tilde{X}_i)-I(X_i;\hat{X}_i^0)\nn\\
 &=& I(X_i;\hat{X}_i|\hat{X}_i^0)\nn\\
 &\geq& R_{X_i|\hat{X}_i^0}(D_i)\nn\\
 &\geq& R_{X_i|W}(D_i)\label{eq 60},\eqa
 where (\ref{eq 60}) is from the Markov chain $X_i-W-\hat{X}_i^0$.
Using (\ref{eq 60}), we have
 \bqn R_{X_1}(D_1)+R_{X_2}(D_2)-I(X_1;X_2)
 &\geq&R_{X_1|W}(D_1)+I(X_1;W)+R_{X_2|W}(D_1)+I(X_2;W)-I(X_1;X_2)\\
 &=& R_{X_1|W}(D_1)+R_{X_2|W}(D_2)+I(X_1X_2;W)\\
  &=& R_{X_1X_2|W}(D_1,D_2)+I(X_1X_2;W)\\
 &=& R_{X_1X_2}(D_1,D_2),\eqn
 which completes the proof.

\subsection{Proof of Theorem \ref{th Gaussian}}\label{appendix 6}
First, we will show that the common information of $X_1,X_2$ is only a function of the correlation coefficient  $\rho$.
To show this, let $\tilde{X}_i=\frac{1}{\sigma_i}X_i$, $i=1,2$, thus $\tilde{X}_1,\tilde{X}_2$ are joint Gaussian distributed with
zero mean and covariance matrix
 \bqn \left[
      \begin{array}{cc}
       1& \rho\\
       \rho& 1
      \end{array}
     \right].\eqn
We have the Markov chain that $\tilde{X}_1-X_1-X_2-\tilde{X}_2$ and by the data processing inequality for Wyner's common information
     \cite{Witsenhausen_CI76}, $C(\tilde{X}_1,\tilde{X}_2)\leq C(X_1,X_2)$. On the other hand, we  have the Markov chain that $X_1-\tilde{X}_1-\tilde{X}_2-X_2$ and  $C(\tilde{X}_1,\tilde{X}_2)\leq C(X_1,X_2)$. Thus, $C(\tilde{X}_1,\tilde{X}_2)= C(X_1,X_2)$.
     Without loss generality, we will consider $\sigma_1^2=\sigma_2^2=1$ in the following.

Let
\bqa X_i=\sqrt{\rho}W+\sqrt{1-\rho}N_i,~ i=1,2,\eqa
where $W,N_1,N_2$ are mutually independent standard Gaussian random variables. It is clear that $X_1,X_2$ are bivariate Gaussian with correlation coefficient $\rho$,
 \bqn C(X_1,X_2)\leq I(X_1,X_2;W)= \frac{1}{2}\log\frac{1+\rho}{1-\rho}.\eqn
 Next we will show that \bqn C(X_1,X_2)\geq \frac{1}{2}\log\frac{1+\rho}{1-\rho}.\eqn

For any $U$ that satisfies the Markov chain $X_1-U-X_2$, let $D_1$ be the minimum mean square error (MMSE) of estimating $X_1$ using $U$, thus, $D_1=E(X_1-E(X_1|U))^2$. Similarly, let
$D_2=E(X_2-E(X_2|U))^2$. We now show that $I(X_1X_2;U)\geq \frac{1}{2}\log\frac{1+\rho}{1-\rho}$.
\bqa I(X_1X_2;U)&=&H(X_1X_2)-H(X_1|U)-H(X_2|U)\nn\\
&=& I(X_1;U)+I(X_2;U)-I(X_1;X_2)\label{MK 2}\\
&\geq & I(X_1; E(X_1|U))+I(X_2;E(X_2|U))-I(X_1;X_2)\label{MK 3}\\
&\geq & R_{X_1}(D_1)+R_{X_2}(D_2)-I(X_1;X_2)\label{MK 4}\\
&=&\frac{1}{2}\log \frac{1-\rho^2}{D_1D_2},\nn
\eqa
for $D_1\leq 1,D_2\leq 1$, where (\ref{MK 2}) is from the chain rule,
(\ref{MK 3}) is from the Markov chains $X_1-U-E(X_1|U)$, $X_2-U-E(X_2|U)$ and (\ref{MK 4}) is by the definition
of rate distortion function.

Next  we show that $D_1+D_2\leq 2(1-\rho)$, $D_1\leq 1$, $ D_2\leq 1$.
\bqa 2(1-\rho)&=&E(X_1-X_2)^2\nn\\
 &=&E[X_1-E(X_1|U)+E(X_1|U)-X_2]^2\nn\\
&=& E[X_1-E(X_1|U)]^2+E[E(X_1|U)-X_2]^2+2E[(X_1-E(X_1|U))(E(X_1|U)-X_2)]\nn\\
&=&E[X_1-E(X_1|U)]^2+E[E(X_1|U)-X_2]^2\label{MK 5}\\
&=&E[X_1-E(X_1|U)]^2+E[E(X_1|U)-E(X_2|U)+E(X_2|U)-X_2]^2\nn\\
&=&E[X_1-E(X_1|U)]^2+E[X_2-E(X_2|U)]^2+E[E(X_2|U)-E(X_1|U)]^2\nn\\
&&+E[(X_2-E(X_2|U))(E(X_2|U)-E(X_1|U))]\nn\\
&=&E[X_1-E(X_1|U)]^2+E[X_2-E(X_2|U)]^2+E[E(X_2|U)-E(X_1|U)]^2\label{MK 6}\\
&\geq& D_1+D_2\nn
\eqa
where (\ref{MK 5}) is from
\bqn E[(X_1-E(X_1|U))(E(X_1|U)-X_2)]
&=&E[(X_1-E(X_1|U))E(X_1|U)]-E[(X_1-E(X_1|U))X_2]\\
&=&-E[(X_1-E(X_1|U))X_2]\\
&=&-E_{UX_2}[X_2E_{X_1|U}[X_1-E(X_1|U)]]\\
&=&-E_{UX_2}[X_2(E(X_1|U)-E(X_1|U))]=0,\eqn
and (\ref{MK 6}) is from
\bqn &&E[(X_2-E(X_2|U))(E(X_2|U)-E(X_1|U))]\\
&=& E[(X_2-E(X_2|U))E(X_2|U)]-E[(X_2-E(X_2|U))E(X_1|U)]=0\eqn
In addition, we have $D_1=E[X_1-E(X_1|U)]^2=EX_1^2-E[E(X_1|U)^2]\leq EX_1^2=1$.

Thus, \bqn I(X_1X_2;U)&\geq& \frac{1}{2}\log \frac{1-\rho^2}{D_1D_2}\\
&\geq& \frac{1}{2}\log \frac{1-\rho^2}{\left(\frac{D_1+D_2}{2}\right)^2}\\
&\geq&
\frac{1}{2}\log \frac{1-\rho^2}{(1-\rho)^2}\\
&=&\frac{1}{2}\log \frac{1+\rho}{1-\rho}.\eqn

\bibliographystyle{\HOME/tex/IEEEbib}

\end{document}